\documentclass{article}

\usepackage{PRIMEarxiv}

\usepackage[utf8]{inputenc} 
\usepackage[T1]{fontenc}    
\usepackage{hyperref}       
\usepackage{url}            
\usepackage{booktabs}       
\usepackage{amsfonts}       
\usepackage{nicefrac}       
\usepackage{microtype}      
\usepackage{lipsum}
\usepackage{fancyhdr}       
\usepackage{graphicx}       
\graphicspath{{media/}}     

\usepackage{appendix}
\usepackage{amsmath,amsfonts}
\usepackage{algorithmic}
\usepackage{algorithm}
\usepackage{array}
\usepackage[caption=false,font=normalsize,labelfont=sf,textfont=sf]{subfig}
\usepackage{textcomp}
\usepackage{stfloats}
\usepackage{url}
\usepackage{verbatim}
\usepackage{graphicx}
\usepackage{cite}
\usepackage{amssymb}
\usepackage{bm}
\usepackage{bbm}
\usepackage{color}
\usepackage{multirow}
\usepackage{diagbox}
\usepackage{threeparttable}

\newtheorem{definition}{Definition}
\newtheorem{theorem}{Theorem}
\newtheorem{proposition}{Proposition}
\newtheorem{remark}{Remark}
\newtheorem{lemma}{Lemma}
\newtheorem{proof}{Proof}

\title{A Provably Secure Network Protocol for Private Communication with Analysis and Tracing Resistance  
\thanks{This material has been submitted to submitted to the IEEE Journal on Selected Areas in Communications.

This paper was supported by the National Key Research and
Development Program of China (2022YFA1004600), the National
Natural Science Foundation of China (12288201, 12071465).} 
}

\author{
  Chao Ge \\
  Department of Electronic Engineering \\
  Tsinghua University \\
  Beijing\\
  \texttt{gechao@amss.ac.cn} \\
  \And
  Wei Yuan \\
  State Key Laboratory of Mathematical Sciences \\
  Academy of Mathematics and Systems Science, Chinese Academy of Sciences \\
  Beijing\\
  \texttt{wyuan@math.ac.cn} \\
   \And
  Ge~Chen \\
  State Key Laboratory of Mathematical Sciences \\
  Academy of Mathematics and Systems Science, Chinese Academy of Sciences \\
  Beijing\\
  \texttt{chenge@amss.ac.cn} \\
  \And
  Yanbin Pan \\
  State Key Laboratory of Mathematical Sciences \\
  Academy of Mathematics and Systems Science, Chinese Academy of Sciences \\
  Beijing\\
  \texttt{panyanbin@amss.ac.cn} \\
  \And
  Yuan Shen \\
  Department of Electronic Engineering  \\
  Tsinghua University \\
  Beijing\\
  \texttt{shenyuan\_ee@tsinghua.edu.cn} 
}

\begin{document}
\maketitle

\begin{abstract}
Anonymous communication networks have emerged as crucial tools for obfuscating communication pathways and concealing user identities.
However, their practical deployments face significant challenges, including susceptibility to artificial intelligence (AI)-powered metadata analysis, difficulties in decentralized architectures, and the absence of provable security guarantees.
To address these issues, this paper proposes a novel decentralized anonymous routing protocol with resistance to tracing and traffic analysis. 
The protocol eliminates dependencies on the threshold model and trusted third-party setups, ensuring indistinguishable identity privacy even in highly adversarial environments.
Different from traditional empirical security analysis of anonymous networks, this paper rigorously proves indistinguishable identity privacy for users even in extremely adversarial environments.
Furthermore, simulations confirm its practical feasibility, demonstrating both security and efficiency.
By achieving information sharing with privacy preservation, the proposed protocol offers a provably secure solution for privacy-preserving communication in digital environments.
\end{abstract}

\keywords{Anonymous network \and Provable security \and Privacy preservation \and Communication security \and Traffic analysis resistance}

\section{Introduction}\label{sec_intro}
The rapid development of wireless communication technology has driven an exponential increase in terminal devices, intensifying the social reliance on ubiquitous network access.
While facilitating diverse applications, these technologies simultaneously pose significant privacy risks due to the widespread transmission of sensitive data
~\cite{ChenAn:J2023,MakhImra:J2022,ChenZhu:J2022,HitaPagn:J2023}.
The emergence of 6G networks will heighten this concern by generating, storing, and processing vast amounts of data, including precise geolocation tracking and predictive user profiling.
This challenge has spurred urgent research into communication security and privacy preservation.

Traditional cryptographic methods ensure content confidentiality and integrity but are inadequate against emerging threats.
Modern challenges require broader protection that extends to communication metadata, including temporal patterns, frequency characteristics, and relationship dynamics within communication processes~\cite{MakhImra:J2022,ChenZhu:J2022,Partala:J2018,Gilad:J2019,sasy:J2024}.
As artificial intelligence (AI) reshapes cyberattack strategies, adversaries are increasingly targeting metadata inference over content theft~\cite{ChenAn:J2023,ChenZhu:J2022,Hoang:J2024}.
By leveraging advanced data analytics and machine learning, they extract behavioral patterns from communication processes to deduce critical and sensitive information~\cite{Siri:C2021,Nguy:J2021}.
This is particularly vital in next-generation networks, where metadata can expose operational patterns and user behaviors through sophisticated correlation and AI-driven inference attacks.

To address risks of identity traceability and data linkability, anonymous communication networks have become a key research focus~\cite{Zhang:C2005,Edman:J2009,shir:J2018,GolReeSyv:C96,DinMatSyv:C04,JohnChris:C2013,LingLuo:J2012,Murdoch:C2007,SerSew:C2003,ZanHar:C2011,EggSch:C2013,HerGro:C2011,Hoang:C2018,Jeong:C2016,Danezis:B2009,PioHay:B2017,Diaz:B2010, Kesdogan:B1998}.
Anonymous communication networks enable users to communicate without revealing their identities, locations, or behavioral patterns, thereby enhancing confidentiality.
They are applied in secure military operations, privacy-preserving networks, and anonymous social platforms that promote free expression. Also, they play a vital role in e-commerce by preventing third-party tracking, as well as in e-democracy, online surveys, where anonymity ensures impartiality and data authenticity.
These networks balance information sharing with privacy preservation, with prominent examples including the onion routing (Tor), invisible internet project (I2P), and Nym Mixnet.
\textit{Tor} is a connection-oriented system using multi-hop proxy sequences of volunteer nodes~\cite{GolReeSyv:C96,DinMatSyv:C04}.
It offers low latency, high anonymity, and ease of deployment, making it the most widely used anonymous network~\cite{JohnChris:C2013}.
However, due to its lack of traffic obfuscation, it is susceptible to activity pattern detection and de-anonymization via website fingerprinting and end-to-end correlation attacks~\cite{JohnChris:C2013,LingLuo:J2012,Murdoch:C2007}.
\textit{I2P} is a hidden network that employs garlic routing, a variant of onion routing, with one-way encryption for end-to-end communication~\cite{SerSew:C2003}.
Using short-lived links, it reduces third-party tracking risks~\cite{ZanHar:C2011} and replaces Tor's centralized directories with a Distributed Hash Table (DHT), eliminating reliance on a central authority.
However, secure DHT design remains challenging~\cite{Chaum:B1983}, and I2P is vulnerable to de-anonymization by global adversaries conducting traffic analysis~\cite{EggSch:C2013,HerGro:C2011,Hoang:C2018,Jeong:C2016}.
Built on the Loopix protocol,
\textit{Nym Mixnet} offers superior metadata protection through cryptographic reordering and independent message routing~\cite{Chaum:J1981,Danezis:B2009,PioHay:B2017}.
Despite challenges in bandwidth, computation, and latency, Mixnets are critical for privacy-preserving communication~\cite{Diaz:B2010,Kesdogan:B1998}. The evolving Nym network extends to a universal incentivized Mixnet for anonymous email and messaging.


Among these networks, preserving the anonymity of routing information remains a fundamental security challenge.
These networks mainly depend on server-router interactions. However, bidirectional interactions introduce security vulnerabilities.
Shi and Wu~\cite{ShiWu:C21} propose the Non-Interactive Anonymous Router (NIAR) scheme.
By eliminating interaction-related risks, NIAR ensures security even in the presence of untrusted nodes, a key capability that motivates our work.
We defer the brief introduction of NIAR to Subsection~\ref{subsec_niar}.

\textbf{Motivations:}
In summary, current anonymous routing systems, particularly interactive protocols, provide substantial privacy benefits by effectively obfuscating communication paths, thereby preventing adversaries from identifying the source or destination of messages.
However, the security of these anonymous routing systems relies on the assumption that a majority of routing nodes remain uncompromised, known as the \textit{threshold model}~\cite{ChauMay:C2019,FeigeJohn:J2012}.
As a result, these protocols guarantee anonymity only under this strict condition; if an adversary compromises a significant fraction of nodes, privacy protections weaken.
The NIAR scheme, operating non-interactively, is built on a well-defined mathematical security model that enables theoretical proofs and provides a solid foundation for security analysis.
Nevertheless, this scheme relies on a trusted initial setup, posing practical challenges in fully decentralized settings.
Additionally, its computational complexity scales quadratically with the number of participants, leading to inefficient routing computation.

Motivated by these limitations, this paper designs a decentralized private communication network protocol that eliminates reliance on both the threshold model and a trusted initial setup.
The protocol is provably secure, with its resistance to analysis and tracing verified through formal proofs.
Notably, the mathematically grounded security ensures resilience against evolving adversarial strategies, as its rigorously verified cryptographic invariants remain robust against future attacks.
By overcoming the limitations of existing systems, our protocol establishes a novel scheme for privacy-preserving communication in fully decentralized and untrusted environments, offering a theoretically sound and practically feasible approach.

\textbf{Contributions:}
For the significant but challenging private and covert communications, we design a decentralized and threshold-model-free protocol to conceal communication paths and identities of communicating parties. The main contributions are outlined as follows.
\begin{itemize}
    \item
    We propose a decentralized mechanism to dynamically generate routing configurations in the initial setup. This approach overcomes the critical reliance on a trusted initial setup inherent in NIAR schemes and enables practical operation in fully decentralized environments, thereby facilitating adaptability to dynamic network conditions.
    \item
    We present a detailed implementation of a decentralized anonymous routing protocol that eliminates reliance on the threshold model and any trusted initial setup.
    This protocol ensures indistinguishable identity privacy even in extremely adversarial environments.
    This marks a significant advancement over existing systems constrained by strict trust requirements.
    \item     Provable security has long been a fundamental challenge in anonymous networks, remaining unresolved due to inherent cryptographic complexity. Different from traditional empirical security analysis of anonymous networks, this paper rigorously proves indistinguishable identity privacy for users even in extremely adversarial environments.
This result establishes a formidable defense against both present and future adversarial strategies, offering a level of assurance unattainable by empirically driven approaches.
\end{itemize}
Simulation results confirm the practical applicability of the proposed protocol, demonstrating that it is both theoretically secure and practically efficient.

\textbf{Structure:}
Section \ref{sec_pre} presents the preliminaries of this paper, including the application scenario, the building blocks of our protocol, and the hardness assumptions for the theoretical security evaluation.
Section \ref{sec_protocol} introduces the decentralized anonymous communication protocol, which operates independently of any threshold model.
Subsequently, the theoretical security results of the protocol are detailed in Section \ref{sec_eva}.
Section \ref{sec_sim} provides simulation results and a detailed analysis to demonstrate the protocol’s practical efficiency.
Finally, conclusions are drawn in Section \ref{sec_conclu}.

\textbf{Notation:}
Throughout this paper,
we use $\lambda \in \mathbb{N}^+$ to denote the security parameter which measures the input size of the computational problem. By convention, the security parameter is input in unary form, denoted as $1^{\lambda}$, representing a string consisting of $\lambda$ ones.
The notation $\mathtt{poly}(\lambda)$ denotes a \textit{polynomial function} in $\lambda$, meaning that its growth rate is bounded by some polynomial in $\lambda$.
The notation $\mathtt{negl}(\lambda)$ denotes a \textit{negligible function} in $\lambda$.
A function is considered negligible if it decreases more rapidly than the inverse of any polynomial as $\lambda$ grows, a property critical for ensuring that certain probabilities become vanishingly small as the security parameter increases.
For a given prime $q$,
let $\mathbb{Z}_q$ denote the ring of integers modulo $q$, which is the set $\{0, 1, \dots, q-1\}$ equipped with addition and multiplication operations modulo $q$.
Let $\mathbb{Z}_q^{\times}$ denote the multiplicative group of units of $\mathbb{Z}_q$, consisting of all non-zero elements $\{1, 2, \dots, q-1\}$, since these elements are invertible when $q$ is prime.
Moreover, $\mathbb{Z}_q^{n\times m}$ denotes the set of all $n \times m$ matrices with elements from $\mathbb{Z}_q$.
As introduced in \cite{Groth:C2008}, let $\{q, \mathbb{G}_1, \mathbb{G}_2, \mathbb{G}_T, g_1, g_2, g_T, \mu\}$ denote a bilinear group, where: i) $\mathbb{G}_1$, $\mathbb{G}_2$, and $\mathbb{G}_T$ are cyclic groups of order $q$, whose generators are $g_1$, $g_2$, and $g_T$, respectively; ii) $\mu: \mathbb{G}_1 \times \mathbb{G}_2 \to \mathbb{G}_T$ is a non-degenerate bilinear map, and $g_T = \mu(g_1, g_2)$.
For $t \in \{1, 2, T\}$, we define the notation $[\![a]\!]_t = g_t^a \in \mathbb{G}_t$. For convenience, we use $[\![a]\!]_{1,2}$ to represent the pair $([\![a]\!]_1,[\![a]\!]_2)$.
For vectors, let $\boldsymbol{\xi} \in \mathbb{Z}_q^{1 \times \ell}$ be a row vector and $\boldsymbol{\beta} \in \mathbb{Z}_q^{\ell \times 1}$ be a column vector.
We denote $[\![\boldsymbol{\xi}]\!]_1 = ([\![\xi_1]\!]_1, \dots, [\![\xi_\ell]\!]_1)$, a vector of group elements in $\mathbb{G}_1$,
$[\![\boldsymbol{\beta}]\!]_2 = ([\![\beta_1]\!]_2, \dots, [\![\beta_\ell]\!]_2)^T $, a vector of group elements in $\mathbb{G}_2$.
Using the bilinear map $\mu$, the inner product is computed as $[\![\langle \boldsymbol{\xi}, \boldsymbol{\beta} \rangle]\!]_T = \mu([\![\boldsymbol{\xi}]\!]_1, [\![\boldsymbol{\beta}]\!]_2) \in \mathbb{G}_T$, which is also written as $[\![\boldsymbol{\xi}]\!]_1 [\![\boldsymbol{\beta}]\!]_2$.
The operator $\leftarrow$ represents the random sampling of an output of a randomized algorithm, $\overset{{\scriptscriptstyle\$}}{\leftarrow}$ represents the uniform sampling of an element from a set.
We operate within a standard computational model and define an adversary $\mathcal{A}$ as \textit{probabilistic polynomial time} (p.p.t.), meaning that $\mathcal{A}$ executes in polynomial time relative to the security parameter.
Unless otherwise specified, all algorithms are probabilistic. The probability that $\mathcal{A}$ outputs $1$ in an experiment $\textbf{Exp}$ is denoted by $\mathrm{Pr}\left[ 1 \leftarrow \mathcal{A}\left(\textbf{Exp}\right) \right]$.
For two bit-strings $\boldsymbol{x}$ and $\boldsymbol{y}$, their concatenation is represented as $\boldsymbol{x} \| \boldsymbol{y}$.

\section{Preliminaries}\label{sec_pre}
In this section, we first describe the application scenario of our private communication protocol.
Following this, we review the NIAR scheme proposed in~\cite{ShiWu:C21}, and the correlated pseudorandom function (CPRF), which serves as the foundational building block of our protocol.
Finally, we present the relevant hardness assumptions for theoretical security evaluation.

\subsection{Scenario Description}\label{subsec_scen}
To address the privacy-preserving challenges of latency-tolerant remote communication networks, we design a decentralized private communication protocol for asynchronous message transmission.
In this application scenario, the protocol ensures anonymity among all honest senders such that no participant can distinguish the identity of its communication source.
Furthermore, it achieves unlinkability against adversarial entities equipped with comprehensive monitoring capabilities and traffic analysis techniques.

Without loss of generality, we consider the foundational scenario and formalize an $n$-to-$n$ communication paradigm, as illustrated in Fig.~\ref{fig_scenario}.
    \begin{figure*}[t]
    \centering
    \includegraphics[width= 0.9\linewidth]{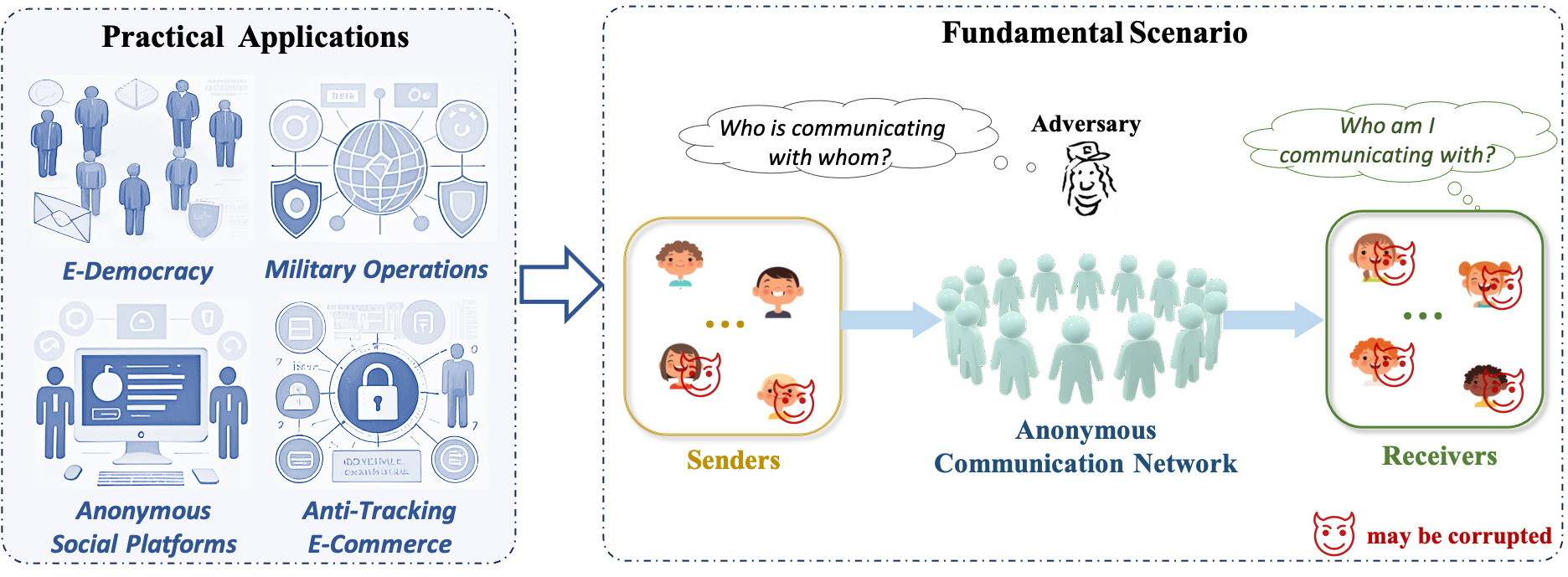}
    \captionsetup{justification=centering}
    \caption{The application scenario of our private communication network protocol.}
    \label{fig_scenario}
    \end{figure*}
The fundamental communication scenario involves $n$ senders and $n$ receivers, where each sender intends to communicate anonymously with a unique receiver. To conceal the identities of participants, all communications undergo a permutation process. Specifically, let $\boldsymbol{\pi} \in S_n$ represent the permutation map between senders and receivers.
For sender $i \in [n]$, let $\boldsymbol{\pi}(i)$ denote the receiver with whom sender $i$ intends to communicate.
This permutation process is executed by an untrusted router, which, by design, is unable to access any knowledge regarding the mapping information $\{i, \boldsymbol{\pi}(i)\}_{i \in [n]}$.
Additionally,  certain senders and receivers may be corrupted, allowing an adversary to gain knowledge of the permutation information known to corrupted participants.
It is assumed that all participants comply with the specified routing protocol, although they may inadvertently leak information.
Our objective is to design private communication networks that can effectively conceal the identities of participants, even when a subset of corrupted participants may collude with the untrusted router.
Let $\mathcal{K}_S \subseteq [n]$ denote the set of corrupted senders and $\mathcal{K}_R \subseteq [n]$ denote the set of corrupted receivers. Let $\mathcal{H}_S := [n] \setminus \mathcal{K}_S$ and $\mathcal{H}_R := [n] \setminus \mathcal{K}_R$ denote the set of honest senders and receivers, respectively.
It is noteworthy that the generation of $\pi$ exhibits permutation invariance with respect to the order of senders.
Without loss of generality, we can assume that $\mathcal{K}_S = \{m+1,\cdots,n\} \subset \{m'+1,\cdots,n\}=\mathcal{K}_R$, where $m'\leq m$, denoting the number of honest receivers and senders, respectively.

\subsection{Non-Interactive Anonymous Router}\label{subsec_niar}
As introduced in~\cite{ShiWu:C21}, the NIAR scheme consists of the following procedures:
\begin{itemize}
    \item $\left( \{\boldsymbol{ek}_i,\boldsymbol{rk}_i\}_{i\in [n]},\boldsymbol{tk} \right) \leftarrow \mathtt{NIAR.Setup}(1^{\lambda},n,\boldsymbol{\pi})$. First, the one-time trusted setup takes the security parameter $1^{\lambda}$, the sender/receiver number $n$, and the routing permutation $\boldsymbol{\pi}$. Subsequently, it outputs the sender keys $\{\boldsymbol{ek}_i\}_{i\in [n]}$, the receiver keys $\{\boldsymbol{rk}_i\}_{i\in [n]}$, and a token $\boldsymbol{tk}$ for the router to encode $\boldsymbol{\pi}$.

    \item $\boldsymbol{ct}_{i,t} \leftarrow \mathtt{NIAR.Enc}\left(\boldsymbol{ek}_i,\boldsymbol{msg}_{i,t},t\right)$. For $t = 1,2,\cdots$, sender $i$ encrypts its message $\boldsymbol{msg}_{i,t}$ with its secret key $\boldsymbol{ek}_i$ and sends ciphertext $\boldsymbol{ct}_{i,t}$ to the router.

    \item $ \{\boldsymbol{ct}'_{i,t}\}_{i\in [n]}\leftarrow \mathtt{NIAR.Rte}\left(\boldsymbol{tk},\{\boldsymbol{ct}_{i,t}\}_{i\in [n]}\right)$. The router uses the token $\boldsymbol{tk}$ to transform the ciphertexts and then forwards $\boldsymbol{ct}'_{i,t}$ to receiver $i$.

    \item $\boldsymbol{msg}_{i,t} \leftarrow \mathtt{NIAR.Dec}(\boldsymbol{rk}_i,\boldsymbol{ct}'_{i,t})$. Receiver $\boldsymbol{\pi}(i)$ uses its secret key $\boldsymbol{rk}_i$ to decrypt $\boldsymbol{ct}'_{i,t}$, thereby retrieving the plaintext $\boldsymbol{msg}_{i,t}$.
\end{itemize}

The indistinguishability security of the NIAR scheme is defined as follows. Let $\mathcal{A}$ denote a non-uniform p.p.t. adversary and $\mathcal{C}$ denote a challenger.
Consider the following experiment $\textbf{NIAR.Exp}^{(b)}(1^{\lambda})$ indexed by $b\in\{0,1\}$:
\begin{itemize}
    \item $n,\mathcal{K}_S,\mathcal{K}_R,\boldsymbol{\pi}_0,\boldsymbol{\pi}_1 \leftarrow \mathcal{A}\left(1^{\lambda}\right)$: $\mathcal{A}$ outputs $\boldsymbol{\pi}_0,\boldsymbol{\pi}_1 \in S_n$ satisfying $\left\{ \left( i,\boldsymbol{\pi}_0(i) \right):i\in \mathcal{K}_S \right\} = \left\{ \left( i,\boldsymbol{\pi}_1(i) \right):i\in \mathcal{K}_S \right\}$, and then sends them to the challenger $\mathcal{C}$.
    \item $\mathcal{C}$ selects $b \overset{{\scriptscriptstyle\$}}{\leftarrow} \{0,1\}$ and runs $\left( \{\boldsymbol{ek}_i,\boldsymbol{rk}_i\}_{i\in [n]},\boldsymbol{tk} \right) \leftarrow \mathtt{Setup}(1^{\lambda},n,\boldsymbol{\pi}_b)$. Then, $\mathcal{C}$ gives $\{\boldsymbol{ek}_i\}_{i \in \mathcal{K}_S}, \{\boldsymbol{rk}_j\}_{j \in \mathcal{K}_R}$, and $\boldsymbol{tk}$ to $\mathcal{A}$.
    \item $\mathcal{A}$ is allowed to make a polynomial number of queries. In the $k$-th query, $\mathcal{A}$ selects and sends two sets of plaintexts $\{\boldsymbol{msg}_{i,t}^0\}_{i \in \mathcal{H}_S}$ and $\{\boldsymbol{msg}_{i,t}^1\}_{i \in \mathcal{H}_S}$ that satisfy $\boldsymbol{msg}_{\boldsymbol{\pi}_0^{-1}(i)}^0 = \boldsymbol{msg}_{\boldsymbol{\pi}_1^{-1}(i)}^1, \forall i\in \mathcal{K}_R \cap \boldsymbol{\pi}_0(\mathcal{H}_S) = \mathcal{K}_R \cap \boldsymbol{\pi}_1(\mathcal{H}_S)$. Then, $\mathcal{C}$ returns $\{\mathtt{Enc}\left(\boldsymbol{ek}_i,\boldsymbol{msg}_{i,t}^b,t\right) \}_{i\in \mathcal{H}_S}$ to $\mathcal{A}$.
\end{itemize}

The NIAR scheme is secure if and only if experiments $\textbf{NIAR.Exp}^{(0)}(1^{\lambda})$ and $\textbf{NIAR.Exp}^{(1)}(1^{\lambda})$ are computationally indistinguishable for any p.p.t. adversary.
Based on the standard Decisional Linear assumption in certain bilinear groups, the NIAR scheme is provably secure~\cite{ShiWu:C21}.
However, it relies on a trusted initial setup, which encodes routing information via a central trusted authority. Therefore, it is incapable of addressing the risk of single-point failure and lacks dynamic adaptability.
In this paper, we aim to enable all participants to collaboratively generate the dynamic routing in a distributed manner without a trusted central node.
Moreover, each participant can access only its own routing information, thereby preventing information leakage and ensuring secure communication as well as the privacy of multi-party participation.

\subsection{Correlated pseudorandom Function}\label{subsec_cprf}
Here we introduce the Correlated Pseudorandom Function (CPRF), which serves as an essential building block of our protocol.
Following \cite{Boyle:C2020,ShiWu:C21}, we use the CPRF as a random number generator.
For distinct periods $t=0,1,\cdots$, each sender $i \in [n]$ independently computes a secret $K_i(t)$ by running the CPRF, which is parameterized by the security parameter $1^\lambda$, the number of senders $n$, and a prime $q$,
denoted by
$$(K_1(t),\cdots,K_n(t)) \leftarrow \mathtt{CPRF}(1^{\lambda},n,q),$$
where $K_1(t),\cdots,K_n(t)$ satisfy the following conditions:
\begin{enumerate}
    \item[C1.]  $K_i(t) \in \mathbb{Z}_q$;\label{cond:1}
    \item[C2.] $\sum_{i} K_i(t) = 0$;\label{cond:2}
    \item[C3.] for any non-uniform p.p.t. adversary $\mathcal{A}$,
    \begin{equation*}
        \Big|\Pr\left[ 1 \leftarrow \mathcal{A}\left(\textbf{CPRF.Exp}^{(0)}\right) \right] 
    - \Pr\left[ 1 \leftarrow \mathcal{A}\left(\textbf{CPRF.Exp}^{(1)}\right) \right] \Big|  = \mathtt{negl}(\lambda),
    \end{equation*}\label{cond:3}
\end{enumerate}
where $\textbf{CPRF.Exp}^{(b)},~b\in\{0,1\}$ is defined as:
\begin{itemize}
    \item $\mathcal{A}$ sends a subset $\mathcal{K}\subset [n]$ with $|\mathcal{K}| \leq n-2$ to the challenger $\mathcal{C}$, implying that there exists at least two trusted senders;
    \item
    $\mathcal{C}$ samples $b \overset{{\scriptscriptstyle\$}}{\leftarrow} \{0,1\}$. After running $\mathtt{CPRF}(1^{\lambda},n,q)$, $\mathcal{C}$ returns $\{K_i\}_{i\in \mathcal{K}}$ to $\mathcal{A}$;
    \item $\mathcal{A}$ submits distinct $t$. If $b = 0$, $\mathcal{C}$ returns $\{K_j(t)\}_{j\notin \mathcal{K}}$. If $b = 1$, $\mathcal{C}$ randomly samples $\{d_{j} \in \mathbb{Z}_{q}\}_{j\notin \mathcal{K}}$ that satisfy $\sum_{j\notin \mathcal{K}}d_{j}=-\sum_{i\in \mathcal{K}}K_i(t)$, then returns $\{d_{j}\}_{j\notin \mathcal{K}}$.
    \item $\mathcal{A}$ outputs $0$ or $1$.
\end{itemize}
\begin{remark}\label{rem_1}
    Condition (C3) guarantees that any non-uniform p.p.t. adversary's view in
    $\textbf{CPRF.Exp}^{(0)}$ and $\textbf{CPRF.Exp}^{(1)}$ are computationally indistinguishable, implying that$\{K_i\}_{i\in \mathcal{K}}$ are computationally indistinguishable from random elements satisfying $\sum_{j\notin \mathcal{K}}d_{j}=-\sum_{i\in \mathcal{K}}K_i(t)$.
\end{remark}

As introduced in~\cite{BonIva:17,AbdBenGay:C19,Gold:J1986}, CPRFs satisfying conditions (C1)-(C3) can be constructed with specific pseudorandom function family (PRF).
A PRF refers to a collection of functions designed to produce outputs that are computationally indistinguishable from truly random outputs when evaluated with a randomly chosen key~\cite{BonIva:17,Gold:J1986}. Formally, let $\{0,1\}^*$ denote the set of all binary strings, representing possible keys, and $\{0,1\}^{\lambda}$ denote the set of binary strings of length $\lambda$. A PRF is defined as a map: $\mathtt{PRF}:\{0,1\}^*\times \{0,1\}^{\lambda} \mapsto \mathbb{Z}_q$.
Let $W_0$ denote the probability that any p.p.t. adversary correctly identifies an output as being generated by $\mathtt{PRF}$, and let $W_1$ denote the probability that the adversary correctly identifies an output as being drawn uniformly at random from $\mathbb{Z}_q$.
A PRF is considered secure if and only if the advantage
$\mathrm{Adv}_{\mathcal{A}}^{\mathtt{PRF}}(\lambda):= |W_0-W_1|$ is negligible.
For $1 \leq i<j \leq n$, $k_{ij}$ are chosen at random.
Let $k_{ij}:=k_{ji}$ for $i>j$.
The derived key $K_i$ is constructed as
\begin{equation}\label{equ_cprf}
    K_i:= \sum_{j\neq i}(-1)^{j<i} \mathtt{PRF}(k_{ij}),
\end{equation}
where $(-1)^{j<i}$ is the indicator function.
As proven in~\cite{BonIva:17,ShiWu:C21}, this construction (\ref{equ_cprf}) satisfies conditions (C1)-(C3)
and the following Lemma \ref{lem_1} holds.
\begin{lemma}\label{lem_1}
    Suppose the pseudorandom function family $\mathtt{PRF}$ is secure. Under construction (\ref{equ_cprf}), it holds that
    \begin{equation*}
    \mathrm{Adv}_{\mathcal{A}}^{\mathtt{CPRF}}(\lambda) :=
    \Big|
    \Pr\left[ 1 \leftarrow \mathcal{A}\left(\textbf{CPRF.Exp}^{(0)}\right) \right] 
    - \Pr\left[ 1 \leftarrow \mathcal{A}\left(\textbf{CPRF.Exp}^{(1)}\right) \right]
    \Big|
    \leq \frac{n(n-1)}{2} \mathrm{Adv}_{\mathcal{A}}^{\mathtt{PRF}}(\lambda).
\end{equation*}
\end{lemma}

By Lemma \ref{lem_1}, we note that if the underlying PRF is secure, then the CPRF constructed by (\ref{equ_cprf}) is also secure.
Since $\mathrm{Adv}_{\mathcal{A}}^{\mathtt{PRF}}$ is negligible and $\frac{n(n-1)}{2}$ is polynomial, we have
$\frac{n(n-1)}{2} \mathrm{Adv}_{\mathcal{A}}^{\mathtt{PRF}}(\lambda)$ is also negligible.

\subsection{External Decisional Linear Assumption}\label{subsec_XDLin}
Different from the NIAR scheme based on the Decisional Linear assumption~\cite{ShiWu:C21}, we construct our protocol on the external decisional linear (XDLin) assumption~\cite{Garg:C2010,Tomi:C2016}, defined as follows.
\begin{definition}[XDLin] \label{def_XDLin}
Let
$\mathtt{pp}_{\mathbb{G}}:=\{q,\mathbb{G}_1,\mathbb{G}_2,\mathbb{G}_T,\mu\} \leftarrow \mathcal{G}(1^\lambda)$
be a bilinear group generated by algorithm $\mathcal{G}(1^\lambda)$.
The external decisional linear assumption holds for $\mathcal{G}$ if and only if the following distributions $P_0$ and $P_1$ are computationally indistinguishable. For $x \in \{0,1\}$,
\begin{itemize}
    \item $P_0:= ([\![a]\!]_{1,2},[\![b]\!]_{1,2},[\![ac]\!]_{1,2},[\![bd]\!]_{1,2},[\![c+d]\!]_{x})$ with $a,b,c,d \overset{{\scriptscriptstyle\$}}{\leftarrow} \mathbb{Z}_q$;
     \item $P_1:= ([\![a]\!]_{1,2},[\![b]\!]_{1,2},[\![ac]\!]_{1,2},[\![bd]\!]_{1,2},[\![e]\!]_{x})$ with $a,b,c$, $d,e \overset{{\scriptscriptstyle\$}}{\leftarrow} \mathbb{Z}_q$,
\end{itemize}
that is, for any p.p.t. $\mathcal{A}$,
\begin{equation*}
    \mathrm{Adv}_{\mathcal{A}}^{\mathrm{XDLin}}(\lambda) := 
    \Big| \mathrm{Pr}\left[ 1 \leftarrow \mathcal{A}(P_0) \right] - \mathrm{Pr}\left[ 1 \leftarrow \mathcal{A}(P_1) \right] \Big| =\mathtt{negl}(\lambda).
\end{equation*}
\end{definition}

The XDLin assumption asserts the hardness of distinguishing a specific element in $\mathbb{G}_T$ from a random element, given certain elements in $\mathbb{G}_1$ and $\mathbb{G}_2$. The advantage $\mathrm{Adv}_{\mathcal{A}}^{\mathrm{XDLin}}(\lambda)$ is negligible.

\section{Design of the Decentralized Anonymous Communication Protocol}\label{sec_protocol}
In this section, we present the design of our threshold-model-free decentralized anonymous communication protocol in two phases (Phase I: Setup; Phase II: Communication).
First, we introduce a decentralized setup method without any trusted authority.
Subsequently, we detail the communication procedures of our anonymous router.
We remark that the protocol operates in a fully decentralized manner and does not rely on threshold models.
To streamline presentation, we present only the plaintext version of the anonymous routing protocol.
This protocol can be easily extended to support encrypted communications through integration with standard cryptographic primitives.
The roadmap of our protocol is illustrated in Fig. \ref{fig_roadmap}.
\begin{figure*}[t]
    \centering
    \includegraphics[width=0.9 \linewidth]{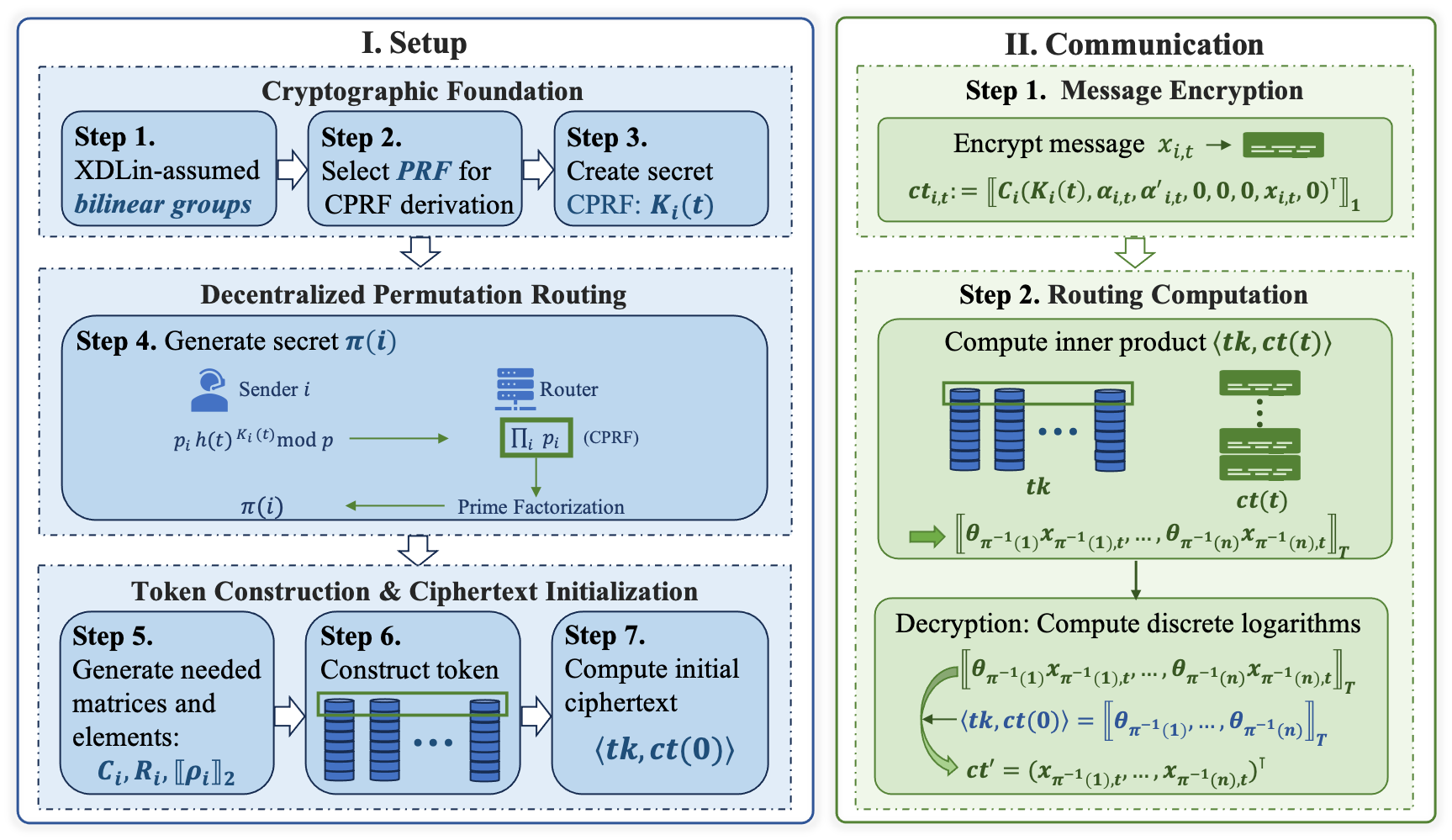}
    \captionsetup{justification=centering}
    \caption{The roadmap of the proposed private communication network protocol.}
    \label{fig_roadmap}
    \end{figure*}

\subsection{Decentralized Setup Process}\label{subsec_setup}
To achieve initialization without any trusted third-party assumptions, the setup process of our protocol involves the following steps.

\textbf{Step 1:} Select a bilinear group that satisfies the XDLin assumption
$\{q,\mathbb{G}_1,\mathbb{G}_2,\mathbb{G}_T,g_1,g_2,g_T,\mu\}$.
Here we note that the generators $g_1$ and $g_2$ randomly selected and undergo dynamic updates in compliance with the security protocol.

\textbf{Step 2:} Select a pseudorandom function family $\mathtt{PRF}:\{0,1\}^*\times \{0,1\}^{\lambda} \mapsto \mathbb{Z}_q$.

\textbf{Step 3:} Select a sufficiently large prime $p$ to support the Diffie-Hellman (DH) protocol \cite{DifHel:J76}.
For $1\leq i<j\leq n$, senders $i$ and $j$ jointly generate the secret key $k_{ij}$ using the DH protocol, whose security is based on the difficulty of solving the discrete logarithm problem.
As described in Fig. \ref{fig_DH},
the DH protocol enables two parties to securely establish a secret key known only to them within an insecure network environment.
To be more specific, senders $i$ and $j$ agree on a public generator $a$.
Then, sender $i$ randomly selects a private key $x_i \overset{{\scriptscriptstyle\$}}{\leftarrow} \mathbb{Z}_p$, computes $a^{x_i} \pmod{p}$, and sends $a^{x_i} \pmod{p}$ to sender $j$.
Subsequently, sender $j$ randomly selects $x_j \overset{{\scriptscriptstyle\$}}{\leftarrow} \mathbb{Z}_p$, computes $a^{x_j} \pmod{p}$, and sends $a^{x_j} \pmod{p}$ to sender $i$. Afterwards, senders $i$ and $j$ can compute the shared secret key $k_{ij}:=a^{x_ix_j} \pmod{p}$.
Let $k_{ij}:=k_{ji}$ for $i>j$.
Hence, each sender $i\in [n]$ can compute $K_i(t):= \sum_{j\neq i}(-1)^{j<i} \mathtt{PRF}(k_{ij},t) \pmod{p}$.
It holds that $\sum_{i=1}^n K_i(t) = 0$.
    \begin{figure}[t]
    \centering
    \includegraphics[width=0.6 \linewidth]{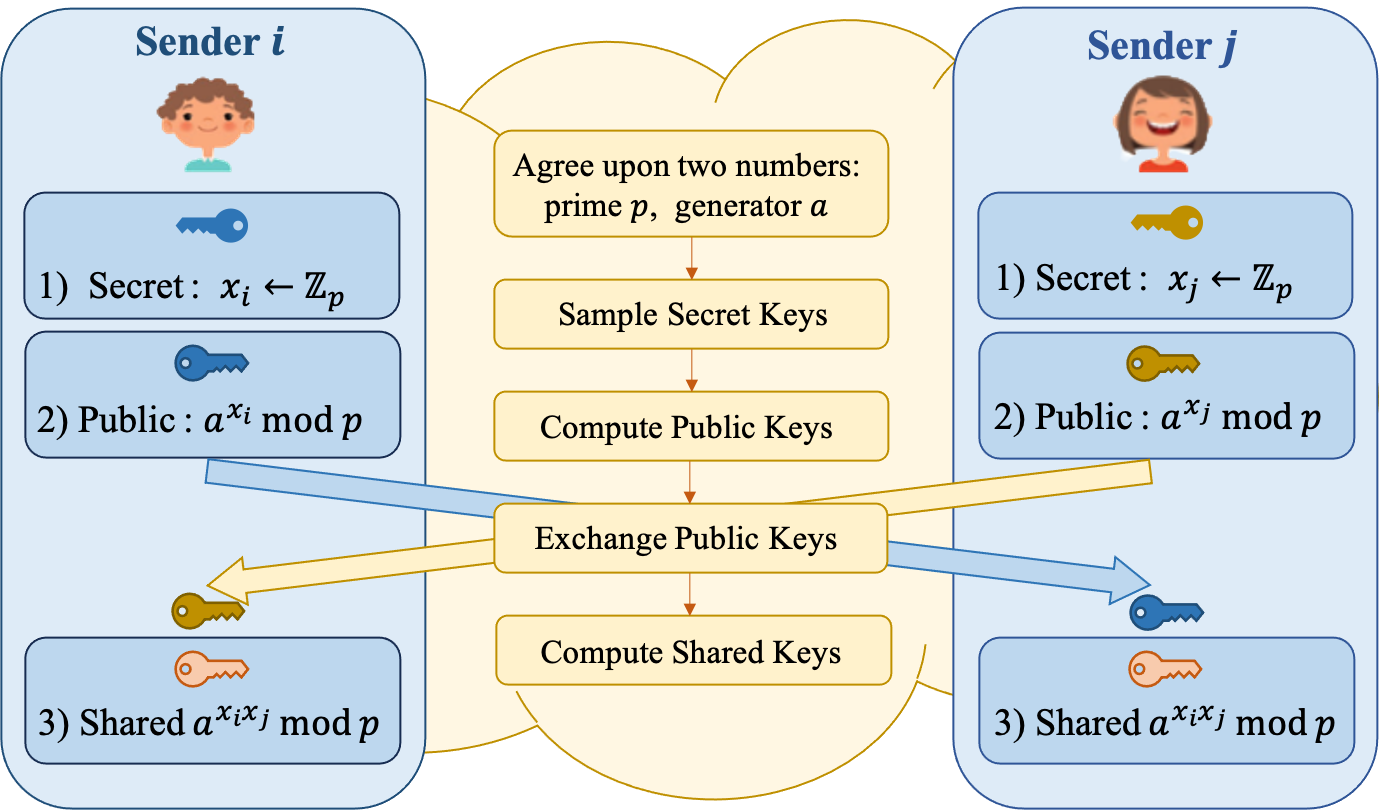}
    \captionsetup{justification=centering}
    \caption{The DH key exchange protocol for senders $i,j$ in Step 3 of the setup process.}
    \label{fig_DH}
    \end{figure}

\textbf{Step 4:} All senders collaboratively generate a permutation $\boldsymbol{\pi} \in S_n$ in a decentralized and anonymous manner, without a trusted third party.
    The detailed process is described in Fig. \ref{fig_routing}.
    Given a public prime list $\mathcal{L}$ and a Hash function $h(t)$, sender $i$ randomly selects a positive integer $r_i \in \mathbb{N}^{+}$.
    Denote the $r_i$-th element of $\mathcal{L}$ by $p_{r_i}$. It is worth mentioning that $p_{r_i}$ should not be excessively large.
    With the public Hash function $h(t)$, sender $i$ computes $p_{r_i}h(t)^{K_i(t)} \pmod{p},$ and then sends the result to the router.
    After receiving $\{p_{r_1}h(t)^{K_1(t)},\cdots,p_{r_n}h(t)^{K_n(t)}\}$, the router computes
    \begin{equation*}
        \prod_{i=1}^n p_{r_i} h(t)^{K_i(t)} \pmod{p} =(\prod_{i=1}^n p_{r_i}) h(t)^{\sum_{i=1}^n K_i(t)} \pmod{p} =\prod_{i=1}^n p_{r_i}.
    \end{equation*}
    Subsequently, the router performs prime factorization on $\prod_{i=1}^n p_{r_i}$, sorts the factors, and publishes the ordered list.
    Each sender $i$ identifies the position of $p_{r_i}$ in the sorted list, denoted as $\boldsymbol{\pi}(i)$.
    Hence, $\boldsymbol{\pi}$ is the desired permutation.
    We note that each sender $i$ knows only $\boldsymbol{\pi}(i)$, and the router has no knowledge of $\boldsymbol{\pi}$.
    \begin{remark}
        If there are repetitions among $p_{r_1},\cdots,p_{r_n}$, repeat the above steps until there are no repetitions.
    \end{remark}
    \begin{figure}[t]
    \centering
    \includegraphics[width= 0.6 \linewidth]{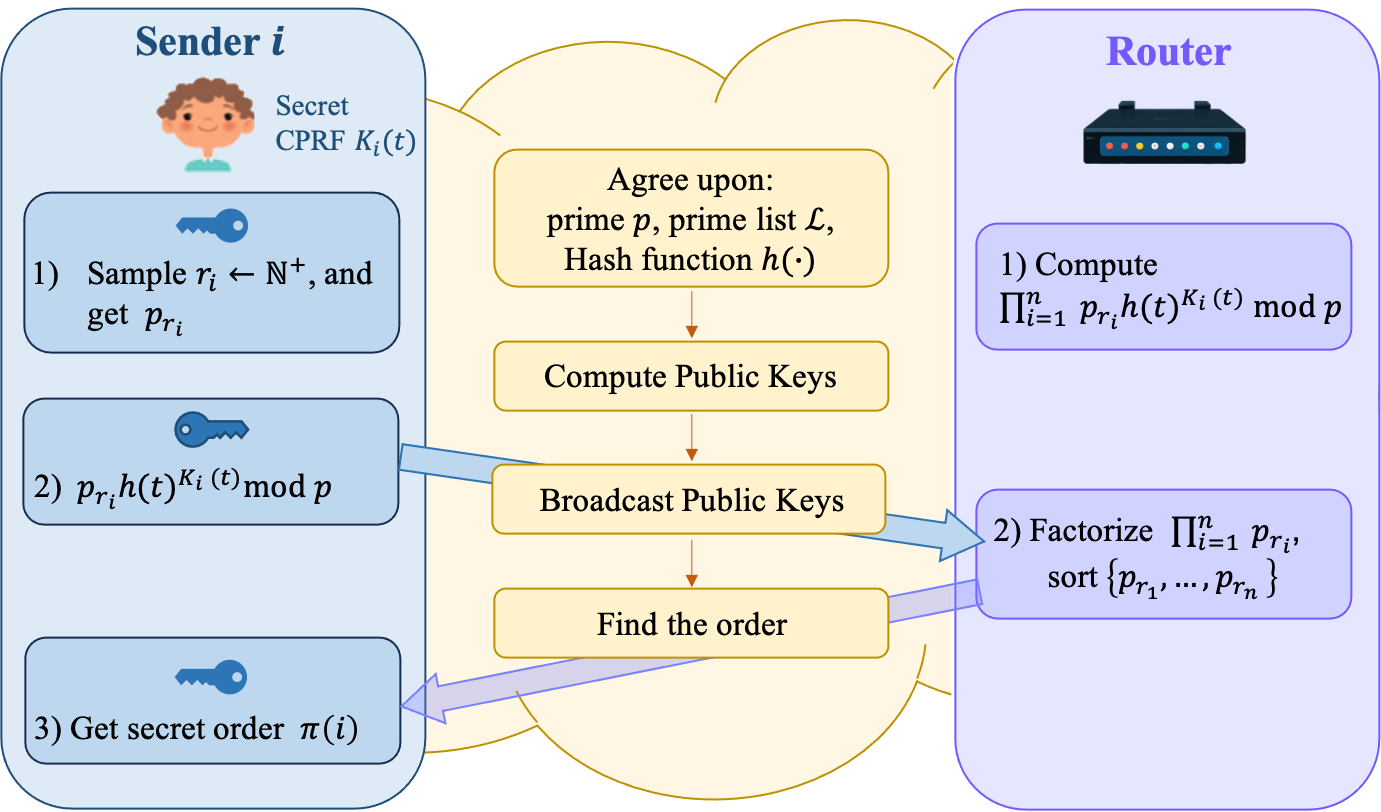}
    \captionsetup{justification=centering}
    \caption{The generation of permutation mapping $\boldsymbol{\pi}$ in Step 4 of the setup process.}
    \label{fig_routing}
    \end{figure}

\textbf{Step 5:} For $i \in [n]$, sender $i$ randomly samples  $\theta_i \overset{{\scriptscriptstyle\$}}{\leftarrow}\mathbb{Z}_q^{\times}$,
and invertible matrix $C_i, R_i \overset{{\scriptscriptstyle\$}}{\leftarrow} (\mathbb{Z}_q)_{8\times 8}^{\times}$ which satisfy 
    $$
        R_i C_i= \mathbf{I}_{8},
    $$
    where $\mathbf{I}_{8}$ is the identity matrix of order $8$.
    Meanwhile, all senders collaboratively generate $\{[\![\rho_1]\!]_2,\cdots,[\![\rho_n]\!]_2\}$, while ensuring that no single sender can dominate or predict the outcome.
    The detailed process is described in Fig. \ref{fig_rho}.
    To be more specific, first, each sender generates a random number $c_i$, which is initially kept confidential. Given the Hash function $h(\cdot)$, each sender computes a commitment to the random number $c_i$, denoted by $h_i=h(c_i)$. The commitment $h_i$ is then publicly disclosed.
    Once all senders have shared their commitments $h_1,\cdots,h_n$, each sender reveals the respective $c_i$.
    Other senders can verify the integrity of the revealed values by checking if $h(c_i)=h_i$, ensuring that $c_i$ has not been altered post-commitment.
    After collecting all verified values $\{c_i\}_{i\in[n]}$, the shared seed $s$ is computed as $s=h(c_1\|\cdots\|c_n)$.
    Utilizing the shared seed $s$ as a foundation, $\{h(s\|i)\}_{i\in[n]}$ are public keys generated by all senders, where $i$ is a unique index corresponding to each sender. Subsequently, $h(s\|i)$ is mapped to $\mathbb{G}_2$ via the standardized Hash-to-Curve algorithm~\cite{icart:C2009hash}, ensuring deterministic and secure encoding of the input data into a valid group element within $\mathbb{G}_2$. Let $[\![\rho_i]\!]_2:= \text{Hash-to-Curve}\left(h(s\|i)\right) \in \mathbb{G}_2$.
    Hash-to-curve is a standardized algorithm for uniformly mapping hash function outputs to points on an elliptic curve. Its security can be formally established based on mathematical foundations. The mapping satisfies computational indistinguishability and provides resilience against chosen-input attacks, where adversarially selected inputs do not affect the statistical distribution of the mapped points.
    Here we remark that $\{[\![\rho_i]\!]_2\}_{i\in[n]}$ are public and collectively generated by all senders, while the values of $\{\rho_i\}_{i\in[n]}$ remain secret.
    \begin{figure}[t]
    \centering
    \includegraphics[width=0.6 \linewidth]{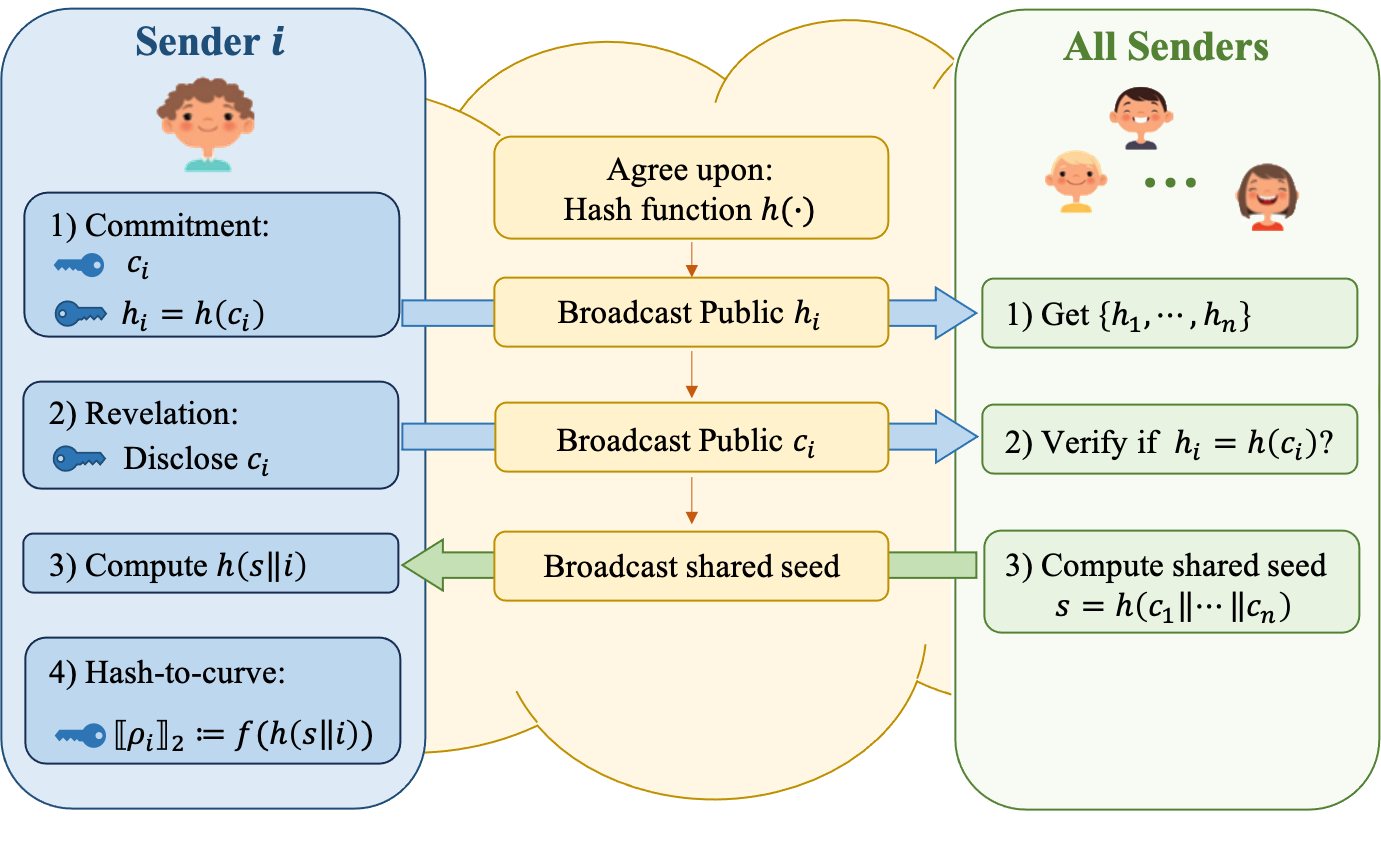}
    \captionsetup{justification=centering}
    \caption{The generation of $\{[\![\rho_1]\!]_2,\cdots,[\![\rho_n]\!]_2\}$ in Step 5 of the setup process.}
    \label{fig_rho}
    \end{figure}

\textbf{Step 6:} For $i \in [n]$, sender $i$ computes $n$ tokens. The $j$-th token corresponding to sender $i$ is computed by
    \begin{equation}\label{equ_tk}
        \boldsymbol{tk}_{i}^j:= [\![\left(\rho_j,0,0,\beta_{i,j},\gamma_{i,j},0,  \theta_i\delta_{j,\boldsymbol{\pi}(i)},0\right)R_i]\!]_{2},
    \end{equation}
    where $\beta_{ij},\gamma_{ij} \overset{{\scriptscriptstyle\$}}{\leftarrow} \mathbb{Z}_q$,
    $\delta_{j,\boldsymbol{\pi}(i)} = 1$ if $j=\boldsymbol{\pi}(i)$,
    and otherwise $\delta_{j,\boldsymbol{\pi}(i)} = 0$.
    Then, sender $i$ sends $\{\boldsymbol{tk}_{i}^1,\cdots,\boldsymbol{tk}_{i}^n\}$ to the router.
    After receiving
    $$\begin{pmatrix}
       \boldsymbol{tk}_{1}^1\\\vdots\\\boldsymbol{tk}_{1}^n
      \end{pmatrix},
      \begin{pmatrix}
       \boldsymbol{tk}_{2}^1\\\vdots\\\boldsymbol{tk}_{2}^n
      \end{pmatrix},\cdots
      \begin{pmatrix}
       \boldsymbol{tk}_{n}^1\\\vdots\\\boldsymbol{tk}_{n}^n
      \end{pmatrix},
      $$
      the router reassembles them row by row.
      Let
      $$\boldsymbol{tk}^j:= \left( \boldsymbol{tk}_{1}^j,\cdots,\boldsymbol{tk}_{n}^j \right)$$ denote the $j$-th row, which is a row vector of length $8n$.
      Therefore, we can obtain the entire routing token
      $$
        \boldsymbol{tk}:= \left( \boldsymbol{tk}^1, \cdots, \boldsymbol{tk}^n \right)^{\top}.
      $$

\textbf{Step 7:} Let $\alpha_{i,0}, \alpha'_{i,0} \overset{{\scriptscriptstyle\$}}{\leftarrow} \mathbb{Z}_q$ be random samples. For $i \in [n]$, sender $i$ computes and sends
    $$
        \boldsymbol{ct}_{i,0} := [\![C_i \left(K_i(0), \alpha_{i,0},\alpha'_{i,0},0,0,0,1,0  \right)^{\top} ]\!]_{1}
    $$
    to the router.
    Let $\boldsymbol{ct}(0) := \left(\boldsymbol{ct}_{1,0},\cdots,\boldsymbol{ct}_{n,0}\right)^{\top}$.
    Subsequently, the router computes and stores the inner product $$
    \langle \boldsymbol{tk},\boldsymbol{ct}(0) \rangle
    =\left(
    [\![\theta_{\boldsymbol{\pi}^{-1}(1)}]\!]_{T},
    \cdots,
    [\![\theta_{\boldsymbol{\pi}^{-1}(n)}]\!]_{T}
    \right)^{\top}.
    $$


In summary, the setup process takes the security parameter $1^{\lambda}$ and the number of senders $n$ as input, and outputs
a token $\boldsymbol{tk}$ and an initial ciphertext, which can be denoted by
$$
\left(\boldsymbol{tk},\langle \boldsymbol{tk},\boldsymbol{ct}(0) \rangle \right) \leftarrow \mathtt{Setup}(1^{\lambda},n).
$$

\subsection{Routing Steps in Communication Process}\label{subsec_comm}
Based on the setup process, the routing steps of the communication process can be outlined as follows.

    \textbf{Step 1:} Let $x_{i,t}$ be the $t$-th message that sender $i$ wants to send. With $\alpha_{i,t}, \alpha'_{i,t} \overset{{\scriptscriptstyle\$}}{\leftarrow} \mathbb{Z}_q$,
    sender $i$ computes
    \begin{equation}\label{equ_ct}
    \boldsymbol{ct}_{i,t} := [\![C_i \left(
    K_i(t),\alpha_{i,t},\alpha'_{i,t},0,0,0,x_{i,t},0
    \right)^{\top} ]\!]_{1},
    \end{equation}
    and sends it to the router.
    The router then collects $\{\boldsymbol{ct}_{i,t} \}_{i\in [n]}$
    Let $\boldsymbol{ct}(t) := \left(\boldsymbol{ct}_{1,t},\cdots,\boldsymbol{ct}_{n,t}\right)^{\top}$.

    \textbf{Step 2:} With the routing token $\boldsymbol{tk}$, the router computes the inner product $\langle\boldsymbol{tk},\boldsymbol{ct}(t)\rangle$.
    By equation (\ref{equ_tk}) and equation (\ref{equ_ct}), we have
    $$
     \langle \boldsymbol{tk}_i^j,\boldsymbol{ct}_{i,t} \rangle
    = [\![\rho_j K_i(t) + \theta_i \delta_{j,\boldsymbol{\pi}(i)} x_{i,t} ]\!]_{T}.
    $$
    Hence, we can obtain that
    \begin{align}\label{equ_routing}
    \langle \boldsymbol{tk},\boldsymbol{ct}(t)\rangle \nonumber &=
    \left(
    \sum_{i=1}^n \langle\boldsymbol{tk}_i^1, \boldsymbol{ct}_{i,t}\rangle,
    \cdots,
    \sum_{i=1}^n \langle\boldsymbol{tk}_i^n , \boldsymbol{ct}_{n,t}\rangle
    \right)^{\top} \nonumber \\
    &= \left(
    [\![\theta_{\boldsymbol{\pi}^{-1}(1)} x_{\boldsymbol{\pi}^{-1}(1),t} ]\!]_{T},
    \cdots,
    [\![\theta_{\boldsymbol{\pi}^{-1}(n)} x_{\boldsymbol{\pi}^{-1}(n),t} ]\!]_{T}
    \right)^{\top}.
    \end{align}
    Recall from \textbf{Step $7$} of the setup process that the router has stored an initial ciphertext
    $\langle \boldsymbol{tk},\boldsymbol{ct}(0) \rangle$.
    By computing the discrete logarithm of $[\![ \theta_{\boldsymbol{\pi}^{-1}(i)} x_{\boldsymbol{\pi}^{-1}(i),t} ]\!]_{T}$ with respect to $[\![\theta_{\boldsymbol{\pi}^{-1}(i)}]\!]_{T}$, the router can obtain
    $$ \boldsymbol{ct}'(t):=\left( x_{\boldsymbol{\pi}^{-1}(1),t}, \cdots,x_{\boldsymbol{\pi}^{-1}(n),t} \right)^{\top}.$$

In summary, the routing process is denoted as
$$
\boldsymbol{ct}'(t) \leftarrow \mathtt{Rte}\left(\boldsymbol{tk}, \boldsymbol{ct}_{1,t},\cdots,\boldsymbol{ct}_{n,t}\right).
$$

\section{Theoretical Results of Security}\label{sec_eva}
In this section, we first extend the XDLin assumption to broader scenarios. Following this, we rigorously demonstrate the \textit{provable security} of the proposed protocol through several sequences of \textit{security experiments}.

As stated in Section~\ref{sec_intro}, provable security plays a critical role in establishing the theoretical soundness of security mechanisms in cryptography and secure communication systems.
By providing formal mathematical proofs, provable security provides quantifiable guarantees against specified adversarial models.
Unlike heuristic or empirical security approaches relying on observed attack resistance, provable security guarantees that breaking the protocol would require solving a well-studied computational hardness problem~\cite{Bren:C2021,Roga:C2006}.

Furthermore, provable security frameworks establish precise definitions of security goals, such as indistinguishability, and adversary capabilities, eliminating ambiguous notions of ``security through obscurity"~\cite{Stern:C2003}.
Therefore, provable security not only identifies potential vulnerabilities during the design phase, but also facilitates comparative analysis.
From an evolutionary perspective, it has become indispensable for standardizing cryptographic protocols, as evidenced by its mandatory inclusion in modern algorithm specifications~\cite{Stern:C2003},
making it a cornerstone of protocol design.

\subsection{Theoretical Results Extended by XDLin Assumption}
As stated in Subsection \ref{subsec_XDLin},
the XDLin assumption asserts that an adversary cannot efficiently distinguish between a valid linear combination of group elements and a random one in certain bilinear groups. Building upon this assumption, we can further extend this ``computational indistinguishability" to a broader range of scenarios.
Hence, the following propositions are derived, forming the security foundation for the construction of our protocol.
\begin{proposition}\label{prop_1}
    Assume the XDLin assumption holds for $\mathcal{G}$. For $m=\mathtt{poly}(\lambda) \geq 1$, consider the following distributions:
    \begin{itemize}
    \item $P_0^{(1)} := ([\![a]\!]_{1,2},[\![b]\!]_{1,2},[\![a\boldsymbol{k}]\!]_{1,2},[\![b\boldsymbol{d}]\!]_{1,2},[\![\boldsymbol{k}+\boldsymbol{d}]\!]_{x})$ with $a,b\overset{{\scriptscriptstyle\$}}{\leftarrow} \mathbb{Z}_q$ and $\boldsymbol{k},\boldsymbol{d}\overset{{\scriptscriptstyle\$}}{\leftarrow} \mathbb{Z}_q^{m\times 1}$;
     \item $P_1^{(1)} := ([\![a]\!]_{1,2},[\![b]\!]_{1,2},[\![a\boldsymbol{k}]\!]_{1,2},[\![b\boldsymbol{d}]\!]_{1,2},[\![\boldsymbol{k}'+\boldsymbol{d}]\!]_{x})$ with $a,b\overset{{\scriptscriptstyle\$}}{\leftarrow} \mathbb{Z}_q$ and $\boldsymbol{k},\boldsymbol{k}',\boldsymbol{d}\overset{{\scriptscriptstyle\$}}{\leftarrow} \mathbb{Z}_q^{m\times 1}$.
    \end{itemize}
    For any p.p.t. $\mathcal{A}$, there holds
    $$
    \Big| \mathrm{Pr}\left[ 1 \leftarrow \mathcal{A}(P_0^{(1)}) \right] - \mathrm{Pr}\left[ 1 \leftarrow  \mathcal{A}(P_1^{(1)}) \right] \Big|  \leq m\mathrm{Adv}_{\mathcal{A}}^{\mathrm{XDLin}}(\lambda).
    $$
\end{proposition}

The proof of Proposition \ref{prop_1} is provided in Appendix A.
We remark that Proposition \ref{prop_1} extends the XDLin assumption to the vector case.
Since $\mathrm{Adv}_{\mathcal{A}}^{\mathrm{XDLin}}(\lambda)$ is negligible,
we can obtain that $m\mathrm{Adv}_{\mathcal{A}}^{\mathrm{XDLin}}(\lambda)$ is also negligible.
Hence,
no p.p.t. adversary can distinguish
$P_0^{(1)}$ from $P_1^{(1)}$ with non-negligible advantage.

On this basis, the following Proposition \ref{prop_2} further imposes a sum constraint on the vectors $\boldsymbol{k},\boldsymbol{k}'$.
\begin{proposition}\label{prop_2}
    Assume the XDLin assumption holds for $\mathcal{G}$. For $m=\mathtt{poly}(\lambda) \geq 1$ and any $t \in \mathbb{Z}_q$, consider the following distributions:
    \begin{itemize}
    \item $P_0^{(2)} := ([\![a]\!]_{1,2},[\![b]\!]_{1,2},[\![a\boldsymbol{k}]\!]_{1,2},[\![b\boldsymbol{d}]\!]_{1,2},[\![\boldsymbol{k}+\boldsymbol{d}]\!]_{x})$ with $a,b\overset{{\scriptscriptstyle\$}}{\leftarrow} \mathbb{Z}_q$ and $\boldsymbol{k},\boldsymbol{d}\overset{{\scriptscriptstyle\$}}{\leftarrow} \mathbb{Z}_q^{m\times 1}$;
     \item $P_1^{(2)} := ([\![a]\!]_{1,2},[\![b]\!]_{1,2},[\![a\boldsymbol{k}]\!]_{1,2},[\![b\boldsymbol{d}]\!]_{1,2},[\![\boldsymbol{k}'+\boldsymbol{d}]\!]_{x})$ with $a,b\overset{{\scriptscriptstyle\$}}{\leftarrow} \mathbb{Z}_q$ and $\boldsymbol{k},\boldsymbol{k}',\boldsymbol{d}\overset{{\scriptscriptstyle\$}}{\leftarrow} \mathbb{Z}_q^{m\times 1}$,
    \end{itemize}
    satisfying $\sum_{i=1}^{m}k_i = \sum_{i=1}^{m}k'_i =t$, where $k_i, k'_i$ are the $i$-th elements of $\boldsymbol{k}$ and $\boldsymbol{k}'$, respectively. For any p.p.t. $\mathcal{A}$, there holds
    \begin{equation*}
    \Big| \mathrm{Pr}\left[ 1 \leftarrow \mathcal{A}(P_0^{(2)}) \right] - \mathrm{Pr}\left[ 1 \leftarrow \mathcal{A}(P_1^{(2)}) \right] \Big| \leq (m-1)\mathrm{Adv}_{\mathcal{A}}^{\mathrm{XDLin}}(\lambda).
    \end{equation*}
\end{proposition}

Under the XDLin assumption, for any p.p.t. adversary, the correlation in $P_0^{(2)}$ remains computationally indistinguishable from the randomized structure in $P_1^{(2)}$, even with the sum constraint $\sum_{i=1}^{m}k_i = \sum_{i=1}^{m}k'_i =t$.
The proof in Appendix B demonstrates that the distinguishing advantage is bounded by $(m-1)\mathrm{Adv}_{\mathcal{A}}^{\mathrm{XDLin}}(\lambda)$, which is also negligible.

Based on Proposition \ref{prop_2}, the following Proposition \ref{prop_3} provides $\mathcal{A}$ with an additional query capability.
\begin{proposition}\label{prop_3}
    Assume the XDLin assumption holds for $\mathcal{G}$. For $m=\mathtt{poly}(\lambda) \geq 1$ and any $t \in \mathbb{Z}_q$, consider the following game $(\triangle)$:
    \begin{itemize}
    \item $\mathcal{A}$ receives the distributions $P_0^{(2)}$ and $P_1^{(2)}$ as defined in Proposition \ref{prop_2}.
    \item $\mathcal{A}$ uniformly samples elements $\{r_1,\cdots,r_m\}$ from the group $\mathbb{G}_x$, where $x\in \{1,2\}$, and sends them to the challenger $\mathcal{C}$. Then, $\mathcal{C}$ returns $\{r_1^a,\cdots,r_m^a\}$ to $\mathcal{A}$.
    \end{itemize}
    For any p.p.t. $\mathcal{A}$, it holds that
    \begin{equation*}
    \Big| \mathrm{Pr}\left[ 1 \leftarrow \mathcal{A}^{(\triangle)}(P_0^{(2)}) \right] - \mathrm{Pr}\left[ 1 \leftarrow \mathcal{A}^{(\triangle)}(P_1^{(2)}) \right] \Big|  \leq (m-1)\mathrm{Adv}_{\mathcal{A}}^{\mathrm{XDLin}}(\lambda).
    \end{equation*}
\end{proposition}

The proof of Proposition \ref{prop_3} is presented in Appendix C.
We can observe that Proposition \ref{prop_3} involves additional interactions on the basis of Proposition \ref{prop_2}.
The adversary's ability to query the challenger with elements from
$\mathbb{G}_x$ and receive their images under exponentiation by $a$ simulates limited oracle access.
Under the XDLin assumption, Proposition \ref{prop_3} demonstrates that $\mathcal{C}$'s response does not leak additional information that would aid $\mathcal{A}$ in distinguishing the two distributions $P_0^{(2)}$ and $P_1^{(2)}$.

Furthermore, the following Proposition \ref{prop_4} establishes the equivalence between two interactive games, denoted as $(\triangle)$ and $(\diamond)$.
\begin{proposition}\label{prop_4}
    The games $(\triangle)$ and $(\diamond)$ are the same in procedures except for: in $(\diamond)$
    \begin{itemize}
    \item After $\mathcal{A}$ receives $P_0^{(2)}$ and $P_1^{(2)}$, as defined in Proposition \ref{prop_2}, $\mathcal{C}$ sends the pair $\big\{ (r_1,r_1^{\frac{1}{a}}),\cdots,(r_m,r_m^{\frac{1}{a}})\big\}$ to $\mathcal{A}$, where $\{r_1,\cdots,r_m\}$ are uniformly sampled elements of $\mathbb{G}_x$.
    \end{itemize}
    For any p.p.t. $\mathcal{A}$, it holds that
    \begin{equation*}
    \Big| \mathrm{Pr}\left[ 1 \leftarrow \mathcal{A}^{(\diamond)}(P_0^{(2)}) \right] - \mathrm{Pr}\left[ 1 \leftarrow \mathcal{A}^{(\diamond)}(P_1^{(2)}) \right] \Big|  \leq (m-1)\mathrm{Adv}_{\mathcal{A}}^{\mathrm{XDLin}}(\lambda).
    \end{equation*}
\end{proposition}

The proof of Proposition \ref{prop_4} is presented in Appendix D.
Proposition \ref{prop_4} demonstrates that
the specific modification to $\mathcal{C}$'s behavior in $(\diamond)$ does not affect $\mathcal{A}$'s ability to distinguish between distributions $P_0^{(2)}$ and $P_1^{(2)}$.

\subsection{A Theoretical Result of Indistinguishability Security}
Following~\cite{Bren:C2021,Roga:C2006}, provable security can be established through the formalization of security experiments that test the security of cryptographic schemes. These experiments typically simulate interactions between a challenger and an adversary: the challenger generates keys and responds to the adversary's queries, such as encryption, decryption, or signing, while the adversary attempts to break the scheme under a specific attack model, such as chosen plaintext or ciphertext attacks.
The core of the experiment is to quantify the adversary's probability of distinguishing between experiments and to prove that this advantage is negligible by reducing it to the hardness of a mathematical problem.
These experiments can provide mathematically grounded assurances that the protocol can withstand both theoretical scrutiny and practical cryptanalysis.

In the following, we first define the security experiments for our protocol.
We say that the protocol is indistinguishably secure
if and only if no non-uniform p.p.t. adversary $\mathcal{A}$ can distinguish between the outputs of any two permutations.
To simplify the exposition without loss of generality, we demonstrate that no non-uniform p.p.t. adversary $\mathcal{A}$ can distinguish between any permutation $\boldsymbol{\pi}$ and the identity mapping.
Given that any permutation $\boldsymbol{\pi}$ is indistinguishable from the identity mapping, indistinguishability between any two permutations immediately follows.
Hence, we consider the following two experiments indexed by $b\in \{0,1\}$.
When $b=0$, the permutation $\boldsymbol{\pi}$ is applied.
When $b=1$, the identity mapping is applied.
In $\textbf{Exp}^{(b)}$:
\begin{itemize}
    \item $\mathcal{A}$ interacts with $\mathcal{C}$ to generate a permutation $\boldsymbol{\pi} \in S_m$, which is subsequently extended to a permutation in $S_n$ by setting $\boldsymbol{\pi}(j)=j$ for all $j>m$.
    Then, $\mathcal{A}$ sends $\boldsymbol{\pi}_0 = \boldsymbol{\pi}$, and $\boldsymbol{\pi}_1 = \operatorname{id}:i \to i$ to $\mathcal{C}$.
    \item $\mathcal{C}$ selects $b\overset{{\scriptscriptstyle\$}}{\leftarrow} \{0,1\}$ uniformly at random.
    Then, $\mathcal{C}$ runs the setup algorithm $\left( \boldsymbol{tk},\langle \boldsymbol{tk},\boldsymbol{ct}(0) \rangle \right) \leftarrow \mathtt{Setup}(1^{\lambda},n)$, for $i \in \mathcal{H}_S$,
    \begin{equation*}
        \boldsymbol{tk}_{i}^j = 
        \begin{cases}
        [\![\left(\rho_j,0,0,\beta_{i,j},\gamma_{i,j},0,  \theta_i\delta_{j,\boldsymbol{\pi}(i)},0\right)R_i]\!]_{2}, & b = 0\\
        [\![\left(\rho_j,0,0,\beta_{i,j},\gamma_{i,j},0,  \theta_i\delta_{j,i},0\right)R_i]\!]_{2}, & b = 1
        \end{cases}.
    \end{equation*}
    $\mathcal{C}$ sends $\boldsymbol{tk}$ and $\langle \boldsymbol{tk},\boldsymbol{ct}(0) \rangle$ to $\mathcal{A}$.
    \item  $\mathcal{A}$ makes $Q \left( = \mathtt{poly}(\lambda)\right)$ queries. In the $t$-th query, $\mathcal{A}$ selects a set of plaintexts $\{x_{i,t}\}_{i<m}$ and sends them to $\mathcal{C}$.
    Then, $\mathcal{C}$ computes
    \begin{equation*}
    \boldsymbol{ct}_{i,t}= 
    \begin{cases}
    [\![C_i \left(
    K_i(t),\alpha_{i,t},\alpha'_{i,t},0,0,0,x_{i,t},0
    \right)^{\top} ]\!]_{1}, & b = 0\\
    [\![C_i \left(
    K_i(t),\alpha_{i,t},\alpha'_{i,t},0,0,0,x_{\boldsymbol{\pi}^{-1}(i),t},0
    \right)^{\top} ]\!]_{1}, & b = 1
    \end{cases}.
    \end{equation*}
    and returns $(\boldsymbol{ct}_{1,t},\cdots,\boldsymbol{ct}_{m,t})$ to $\mathcal{A}$.
    \item By checking the obtained ciphertexts, $\mathcal{A}$ outputs $0$ or $1$.
\end{itemize}

By proving that $\textbf{Exp}^{(0)}$ and $\textbf{Exp}^{(1)}$ are computationally indistinguishable, we can conclude the indistinguishability security of our protocol, and the following Theorem \ref{thm_scheme} holds.

\begin{theorem}[Indistinguishability Security]\label{thm_scheme}
Suppose the pseudorandom function family $\mathtt{PRF}$ is secure and the XDLin assumption holds for $\mathcal{G}$, it holds that
    \begin{multline}\label{equ_theo}
        \left| \mathrm{Pr}\left[ 1 \leftarrow \mathcal{A}\left(\textbf{Exp}^{(0)}\right) \right] - \mathrm{Pr}\left[ 1 \leftarrow \mathcal{A}\left(\textbf{Exp}^{(1)}\right) \right] \right| \\
        \leq 2\Big[ \frac{m(m+1)+Qm}{q} + \mathrm{Adv}_{\mathcal{A}}^{\mathrm{CPRF}}(\lambda)+ \left[Q(m^2+m-2)+m^2\right] \mathrm{Adv}^{\mathrm{XDLin}}(\lambda) \Big].\nonumber
    \end{multline}
\end{theorem}
\begin{remark}
     We note that the number of honest senders, $m$, and the number of adversary's queries, $Q$, are both polynomial in $\lambda$, whereas $q$ is typically a large prime that can be exponential in $\lambda$. Hence, the term $\frac{m(m+1)+Qm}{q}$ becomes negligible as $\lambda$ increases.
     For the XDLin-related term $\left[Q(m^2+m-2)+m^2\right] \mathrm{Adv}^{\mathrm{XDLin}}(\lambda)$, since the coefficient $Q(m^2+m-2)+m^2$ is polynomial and $ \mathrm{Adv}^{\mathrm{XDLin}}(\lambda)$ is negligible under the XDLin assumption, the entire term remains negligible.
     Combined with the security of the PRF, $\mathrm{Adv}^{\mathrm{CPRF}}(\lambda)$ is also negligible.
     Therefore, we can conclude that equation (\ref{equ_theo}) holds.
\end{remark}
\begin{remark}
     Theorem \ref{thm_scheme} illustrates that the adversary's distinguishing advantage between $\textbf{Exp}^{(0)}$ and $\textbf{Exp}^{(1)}$ is negligible, indicating that no p.p.t. adversary can effectively distinguish honest senders with non-negligible probability.
     Through formal proofs, we provide a provable security guarantee for the proposed protocol.
     Provable security has long been a fundamental challenge in privacy-preserving communication networks, remaining unresolved due to intrinsic cryptographic complexity.
     Our results fundamentally demonstrate the theoretical feasibility of constructing provably secure anonymous networks, addressing key challenges in simultaneously ensuring computational efficiency and rigorous security certification.
\end{remark}

Next, we present a comprehensive and detailed proof of Theorem \ref{thm_scheme}. The outline of the proof is described in Fig. \ref{fig_proof}.
\begin{figure*}[t]
    \centering
    \includegraphics[width=\linewidth]{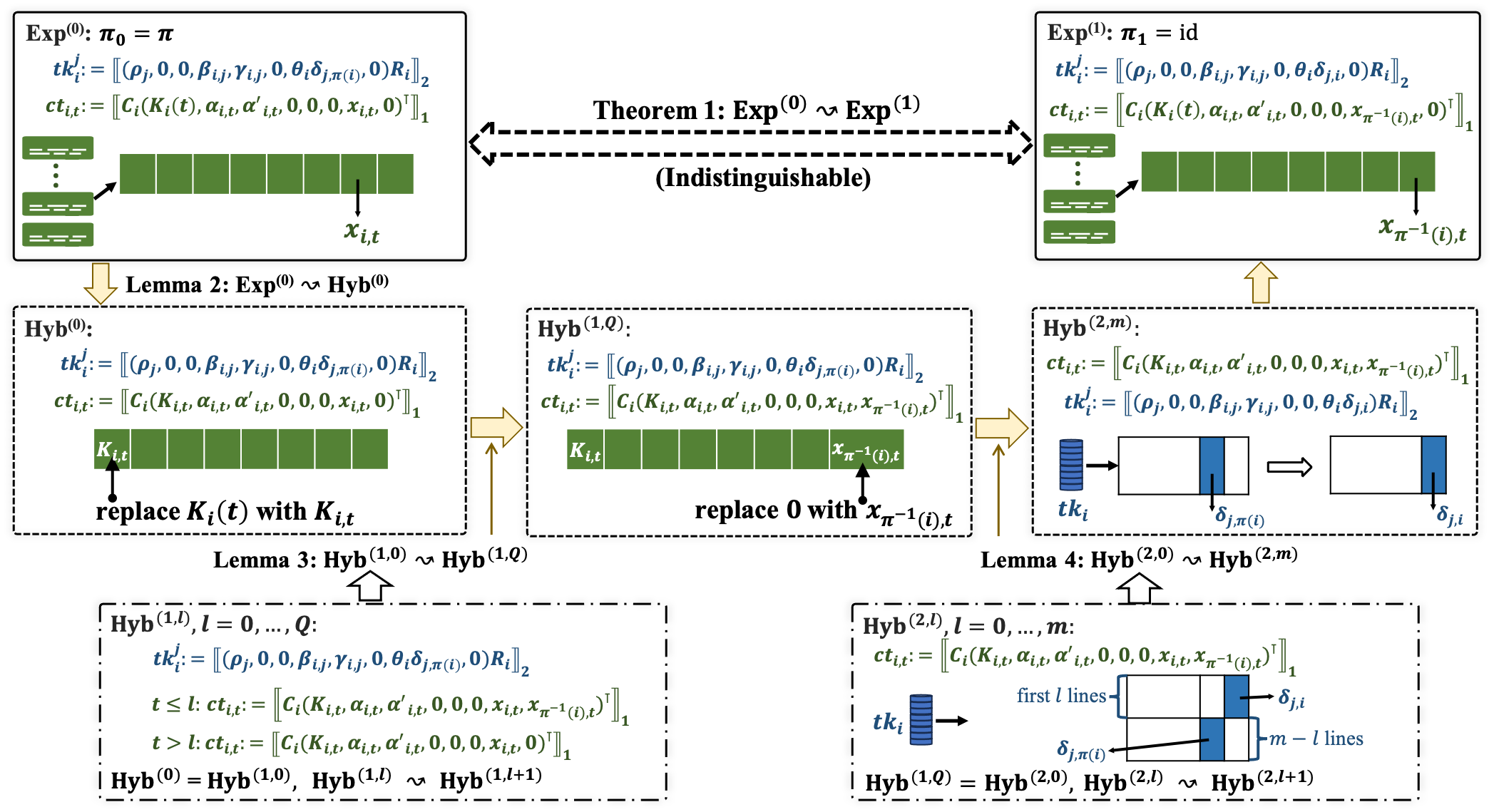}
    \captionsetup{justification=centering}
    \caption{The outline of Theorem \ref{thm_scheme}'s proof.}
    \label{fig_proof}
    \end{figure*}

\begin{proof}[Proof of Theorem \ref{thm_scheme}]
Following the framework of security experiments, we construct several sequences of hybrid experiments to gradually establish a switch between $\textbf{Exp}^{(0)}$ and $\textbf{Exp}^{(1)}$.
By proving that each consecutive pair of hybrid experiments is computationally indistinguishable, we can obtain the indistinguishability of $\textbf{Exp}^{(0)}$ and $\textbf{Exp}^{(1)}$.
We now introduce the detailed hybrid experiments.

\noindent\textbf{(1) Hybrid Experiment $\textbf{Hyb}^{(0)}$:}

First, we define $\textbf{Hyb}^{(0)}$, which is identical to $\textbf{Exp}^{(0)}$ except for the calculation of $(\boldsymbol{ct}_{1,t},\cdots,\boldsymbol{ct}_{m,t})$.
In $\textbf{Hyb}^{(0)}$, when processing the $t$-th $\mathtt{Enc}$ query, the challenger $\mathcal{C}$ replaces
the element $K_i(t)$ in equation (\ref{equ_ct}) with a uniformly random element $K_{i,t} \overset{{\scriptscriptstyle\$}}{\leftarrow} \mathbb{Z}_q$, subject to the constraint
$\sum_{i = 1}^{m}K_{i,t}=\sum_{i = 1}^{m}K_{i}(t).$
This leads to the following Lemma \ref{lem_hyb0}, which shows the indistinguishability of $\textbf{Exp}^{(0)}$ and $\textbf{Hyb}^{(0)}$.
\begin{lemma}\label{lem_hyb0}
Suppose the pseudorandom function family $\mathtt{PRF}$ is secure, then for any non-uniform p.p.t. $\mathcal{A}$, $\textbf{Exp}^{(0)}$ and $\textbf{Hyb}^{(0)}$ are computationally indistinguishable. Formally,
      \begin{equation*}
          \Big| \mathrm{Pr}\left[ 1 \leftarrow \mathcal{A}\left(\textbf{Exp}^{(0)}\right) \right] - \mathrm{Pr}\left[ 1 \leftarrow  \mathcal{A}\left(\textbf{Hyb}^{(0)}\right) \right] \Big|  \leq \mathrm{Adv}_{\mathcal{A}}^{\mathrm{CPRF}}(\lambda).
      \end{equation*}
\end{lemma}
The proof of Lemma \ref{lem_hyb0} can be easily derived from the security of CPRF and Lemma \ref{lem_1} introduced in Subsection \ref{subsec_cprf}.

\noindent\textbf{(2) Hybrid Experiments $\textbf{Hyb}^{(1,\ell)},\ell = 0,1,\cdots,Q$:}

Then,  we define a sequence of hybrid experiments $\textbf{Hyb}^{(1,\ell)}$ for $\ell = 0,1,\cdots,Q.$
The difference between $\textbf{Hyb}^{(1,\ell)}$ and  $\textbf{Hyb}^{(0)}$ lies in the last element in
the calculation of $\boldsymbol{ct}_{i,t}$ in equation (\ref{equ_ct}).
In $\textbf{Hyb}^{(1,\ell)}$,
\begin{itemize}
    \item for queries with $t > \ell$, $\mathcal{C}$ computes $\boldsymbol{ct}_{i,t}$ by
$$
\boldsymbol{ct}_{i,t} =[\![C_i \left(
    K_{i,t},\alpha_{i,t},\alpha'_{i,t},0,0,0,x_{i,t},0
    \right)^{\top} ]\!]_{1};
$$
    \item for queries with $t \leq \ell$, $\mathcal{C}$ computes $\boldsymbol{ct}_{i,t}$ by
$$
\boldsymbol{ct}_{i,t} =[\![C_i \left(
    K_{i,t},\alpha_{i,t},\alpha'_{i,t},0,0,0,x_{i,t},x_{\boldsymbol{\pi}^{-1}(i),t}
    \right)^{\top} ]\!]_{1}.
$$
\end{itemize}

Note that $\textbf{Hyb}^{(1,0)}$ is identical to
 $\textbf{Hyb}^{(0)}$.
The following Lemma \ref{lem_hyb1} illustrates that $\textbf{Hyb}^{(1,Q)}$ and $\textbf{Hyb}^{(1,0)}$ are computationally indistinguishable.
\begin{lemma}\label{lem_hyb1}
Suppose the pseudorandom function family $\mathtt{PRF}$ is secure and the XDLin
assumption holds for $\mathcal{G}$, then for any non-uniform p.p.t. $\mathcal{A}$, $\textbf{Hyb}^{(1,0)}$ and $\textbf{Hyb}^{(1,Q)}$ are computationally indistinguishable. Formally,
      \begin{equation*}
          \left| \mathrm{Pr}\left[ 1 \leftarrow \mathcal{A}\left(\textbf{Hyb}^{(1,0)}\right) \right] - \mathrm{Pr}\left[ 1 \leftarrow \mathcal{A}\left(\textbf{Hyb}^{(1,Q)} \right) \right] \right| 
          \leq 2Q\left[ \frac{m}{q}+(m-1) \mathrm{Adv}^{\mathrm{XDLin}}(\lambda)  \right].
      \end{equation*}
\end{lemma}

The proof of Lemma \ref{lem_hyb1} also employs a sequence of hybrid experiments, with detailed steps provided in Appendix E.

We briefly denote $\textbf{Hyb}^{(1,Q)}$ as $\textbf{Hyb}^{(2)}$.
Subsequently, we define the following sequence of $\textbf{Hyb}^{(2,\ell)}$ for $\ell = 0,1,\cdots,m$.

\noindent\textbf{(3) Hybrid Experiments $\textbf{Hyb}^{(2,\ell)},\ell = 0,1,\cdots,m$:}

The difference between $\textbf{Hyb}^{(2,\ell)}$ and  $\textbf{Hyb}^{(2)}$ lies in the last two elements in
the calculation of $\boldsymbol{tk}_{i}^j$ in equation (\ref{equ_tk}).
In $\textbf{Hyb}^{(2,\ell)}$, for $1\leq i \leq m$,
\begin{itemize}
    \item for $1 \leq j \leq \ell$,
$$\boldsymbol{tk}_{i}^j= [\![\left(\rho_j,0,0,\beta_{i,j},\gamma_{i,j},0,0,\theta_i\delta_{j,i}\right)R_i]\!]_{2};$$
    \item for $\ell+1 \leq j \leq m$,
$$\boldsymbol{tk}_{i}^j= [\![\left(\rho_j,0,0,\beta_{i,j},\gamma_{i,j},0,\theta_i\delta_{j,\boldsymbol{\pi}(i)},0\right)R_i]\!]_{2}.$$
\end{itemize}
Note that $\textbf{Hyb}^{(2,0)}$ is identical to  $\textbf{Hyb}^{(2)}$.
We further demonstrate the indistinguishability of $\textbf{Hyb}^{(2,m)}$ and $\textbf{Hyb}^{(2,0)}$ by the following Lemma \ref{lem_hyb2}.
\begin{lemma}\label{lem_hyb2}
Suppose the pseudorandom function $\mathtt{PRF}$ is secure and the XDLin assumption holds for $\mathcal{G}$, then for any non-uniform p.p.t. $\mathcal{A}$, $\textbf{Hyb}^{(2,0)}$ and $\textbf{Hyb}^{(2,m)}$ are computationally indistinguishable.
Formally,
      \begin{equation*}
        \left| \mathrm{Pr}\left[ 1 \leftarrow \mathcal{A}\left(\textbf{Hyb}^{(2,0)}\right) \right] - \mathrm{Pr}\left[ 1 \leftarrow \mathcal{A}\left(\textbf{Hyb}^{(2,m)}\right) \right] \right| \\
        \leq 2m\left[ \frac{m+1}{q}+(Qm+m-Q) \mathrm{Adv}^{\mathrm{XDLin}}(\lambda)  \right].
      \end{equation*}
\end{lemma}
The proof of Lemma \ref{lem_hyb2} is deferred to the Appendix F.

Combining Lemmas \ref{lem_hyb0}, \ref{lem_hyb1}, and \ref{lem_hyb2},
we can conclude that Theorem \ref{thm_scheme} holds, thereby establishing the indistinguishability of $\textbf{Exp}^{(0)}$ and $\textbf{Exp}^{(1)}$.

\end{proof}

Theorem \ref{thm_scheme} establishes a formal bound on the adversary's distinguishing advantage between $\textbf{Exp}^{(0)}$ and $\textbf{Exp}^{(1)}$.
This bound quantifies the security guarantees of our protocol, demonstrating that the adversary's ability to distinguish the two experiments is negligible, given that the employed PRF is secure and the XDLin assumption holds.
Hence, our protocol's untraceability and unlinkability can be rigorously proven via this indistinguishability result.


\section{Simulations}\label{sec_sim}
In this section, we present some simulations and tests for our proposed protocol.
First, we describe the experimental environment and specific configurations.
Subsequently, 
we provide the communication efficiency evaluation results,
focusing on two critical metrics: encryption/decryption speed and ciphertext transmission efficiency.
Additionally, detailed discussion and analysis of the test results are provided to validate the performance of our protocol.

\subsection{Experimental Configuration}
The simulation experiments were conducted on a virtual machine running Ubuntu $22.04$, equipped with GCC $11.4.0$, $2$ GB of memory, and an Intel i5-8265U processor operating at $1.60$ GHz.
Cryptographic primitives were instantiated using the BLS12-381 curve, a pairing-friendly elliptic curve, which guarantees $128$-bit security under standardized cryptographic assumptions.
All bilinear pairing operations and elliptic curve arithmetic were implemented with the RELIC library, a high-performance cryptographic toolkit optimized for modular exponentiation and pairing-based protocols.

\subsection{Results and Analysis}
The security of our protocol is rigorously demonstrated through formal proofs in Section \ref{sec_eva}.
Since operational performance is crucial for the real-world deployment of cryptographic algorithms, this section focuses on evaluating the communication efficiency of the proposed protocol and demonstrating its practical applicability.
To be more specific, the communication efficiency is assessed through two primary metrics, encryption and decryption speeds, as well as transmission efficiency.
The detailed results are presented below.

\subsubsection{Encryption and Decryption Speeds}
To comprehensively investigate the encryption and decryption performance, two sets of tests are conducted under varying conditions.
\begin{itemize}
    \item Test 1 is designed to evaluate the impact of \emph{message length} and is therefore performed with a \emph{fixed number of users} while \emph{varying the message length}.
    \item Test 2 is designed to assess the impact of \emph{the number of users} and is therefore conducted with a \emph{varying number of users} while \emph{keeping the message length fixed}.
\end{itemize}

\textbf{Test 1:} In this set of tests, the number of users is fixed at $10$, while the message length per communication is varied across $6-10$ bits.
The results are summarized in TABLE~\ref{tab1} and Fig.~\ref{fig_test1}.

\begin{table}[h]
\begin{center}
\caption{Encryption and routing computation time for $10$ users with varying message lengths.}
\label{tab1}
\begin{tabular}{|c|c|c|}
\hline
\textbf{Message Length} & \textbf{User Encryption} & \textbf{Routing Computation} \\
\textbf{(bit)} & \textbf{Time (ms)} & \textbf{Time (ms)} \\
\hline
6& $1.4$ & $780$\\
\hline
7& $1.9$ & $787$\\
\hline
8& $1.3$ & $778$\\
\hline
9& $1.9$ & $794$\\
\hline
10& $1.3$ & $791$\\
\hline
\end{tabular}
\end{center}
\footnotesize
Note: The above results are based on single-threaded benchmark testing. According to the protocol design, multi-threaded implementations are supported. In practical deployments, there is significant potential for improving computational efficiency through methods such as parallel computing, GPU acceleration, and memory expansion..
\end{table}

\begin{figure}[t]
    \centering
    \includegraphics[width= 0.8\linewidth]{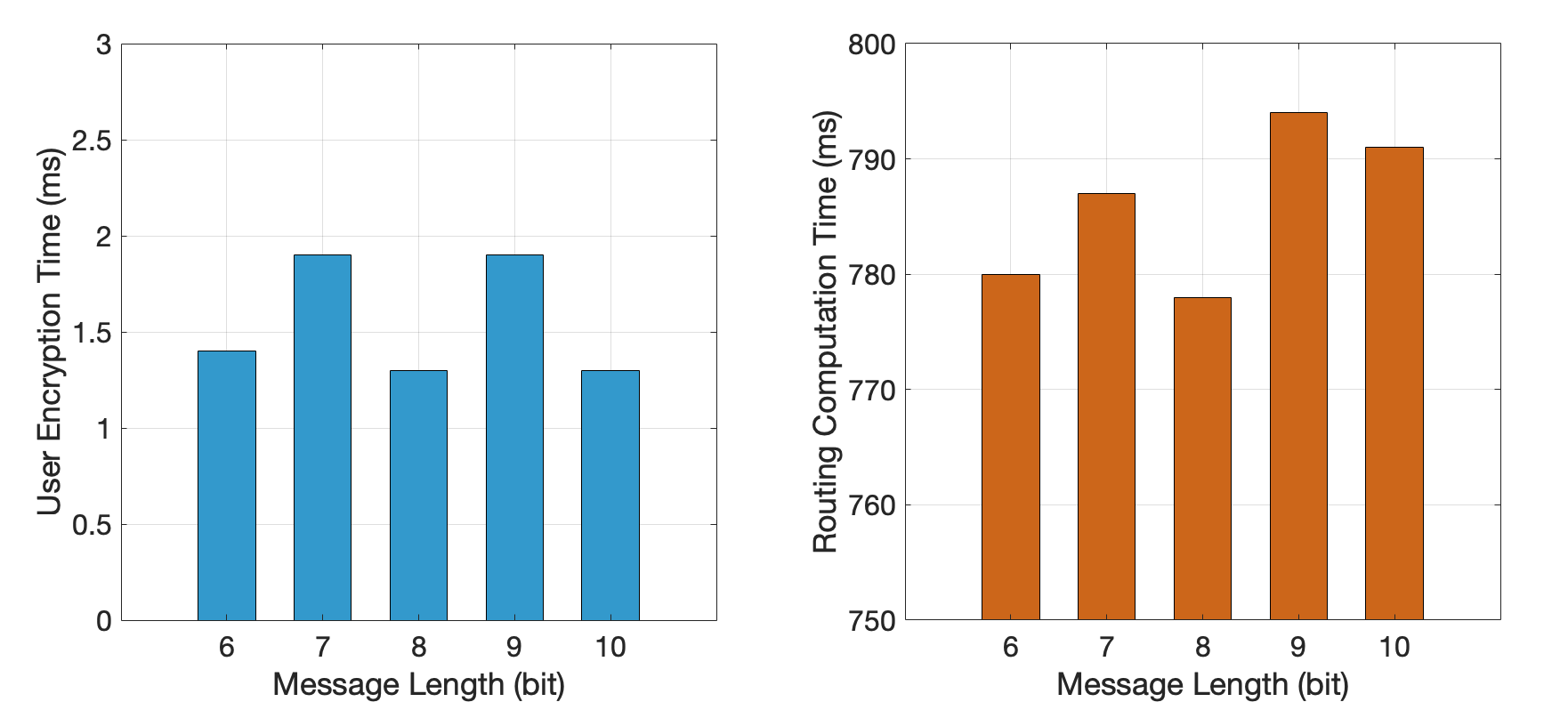}
    \captionsetup{justification=centering}
    \caption{Encryption and routing computation time for $10$ users with varying message lengths.}
    \label{fig_test1}
\end{figure}

As observed from TABLE \ref{tab1} and Fig.~\ref{fig_test1},
the user encryption time exhibits minor fluctuations, ranging from $1.3$ to $1.9$ ms, indicating that it is largely independent of message length within this range.
Similarly, the routing computation time exhibits only minor fluctuations between $778$ and $791$ ms, suggesting stability across the tested message lengths.
These results demonstrate that the computational cost associated with encryption and routing does not scale significantly with small increases in message length.
Furthermore, a larger routing memory, capable of storing more extensive exponent tables for $g_T^{\theta_i}$,
could facilitate the transmission of longer messages per communication, potentially improving efficiency.

\textbf{Test 2:} In the second set of tests, the message length is fixed at $8$ bits per user, while the number of users varies from $5$ to $25$. The results are presented in TABLE \ref{tab2} and Fig.~\ref{fig_test2}.

\begin{table}[h]
\begin{center}
\caption{Encryption and routing computation time for varying numbers of users with a fixed $8$-bit message length.}
\label{tab2}
\begin{tabular}{|c|c|c|}
\hline
\textbf{User Number} & \textbf{User Encryption} & \textbf{Routing Computation} \\
 & \textbf{Time (ms)} & \textbf{Time (ms)} \\
\hline
5& $1.5$ & $204$\\
\hline
10& $1.3$ & $778$\\
\hline
15& $1.6$ & $1733$\\
\hline
20& $1.6$ & $3080$\\
\hline
25& $1.6$ & $4833$\\
\hline
\end{tabular}
\end{center}
\footnotesize
Note: The above results are based on single-threaded benchmark testing. According to the protocol design, multi-threaded implementations are supported. In practical deployments, there is significant potential for improving computational efficiency through methods such as parallel computing, GPU acceleration, and memory expansion.
\end{table}

\begin{figure}[t]
    \centering
    \includegraphics[width=0.8 \linewidth]{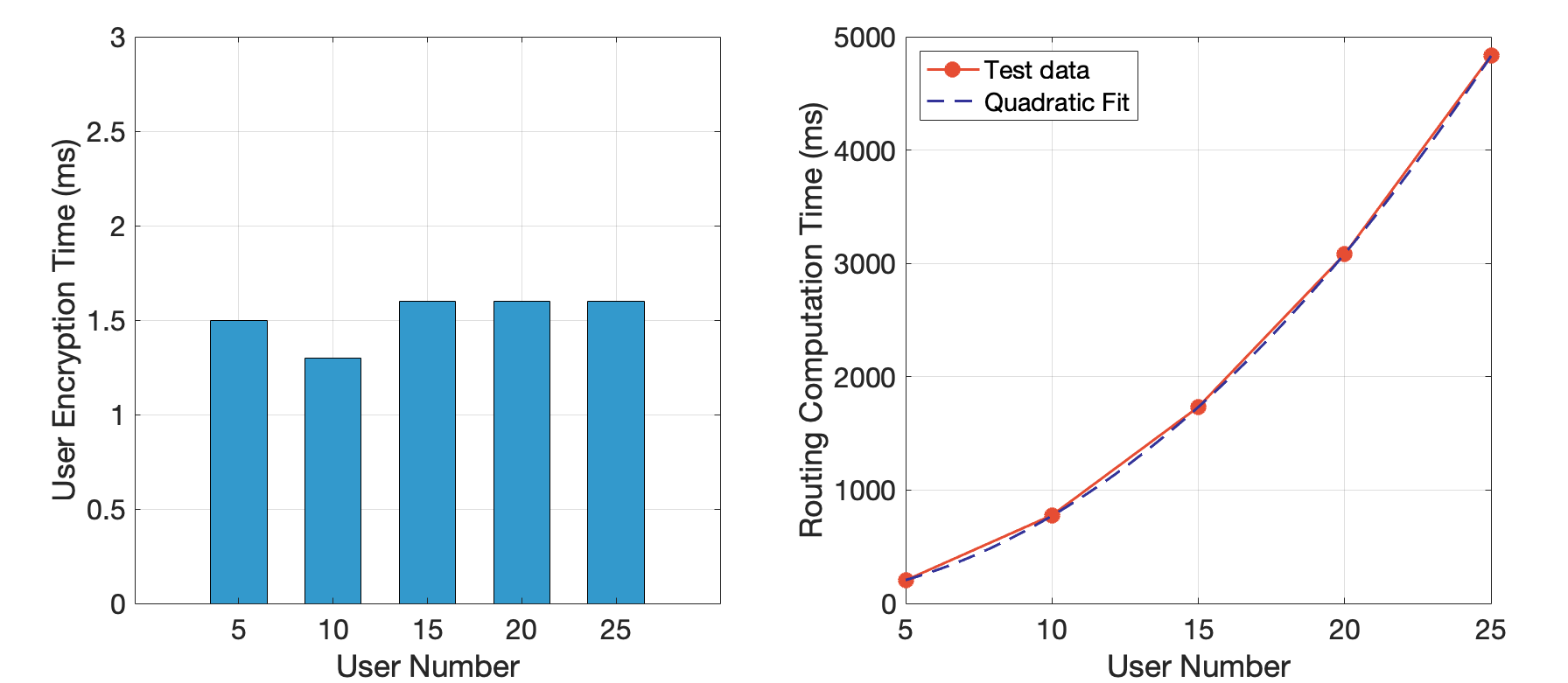}
    \captionsetup{justification=centering}
    \caption{Encryption and routing computation time for varying numbers of users with a fixed $8$-bit message length.}
    \label{fig_test2}
\end{figure}

From TABLE~\ref{tab2} and Fig.~\ref{fig_test2}, we can similarly conclude that the encryption time per user is observed to remain stable, ranging between $1.3$ and $1.6$ ms, further confirming its independence from the number of users.
In contrast, the routing computation time exhibits polynomial growth with an increasing number of users, rising from $204$ ms for $5$ users to $4833$ ms for $25$ users.
This suggests that while individual encryption operations are efficient and scalable, the routing process becomes a performance bottleneck as the system scales to support more users. This issue similarly exists in the NIAR scheme.
Nevertheless, the decentralized design of our protocol allows this limitation to be mitigated through integration with an efficient grouping strategy, which will be the focus of our future work.

Compared to the NIAR scheme in~\cite{ShiWu:C21}, given that each sender transmits a message of $b$ bits and each message spans $l$ bytes, our protocol requires $4nb/l$ group multiplications during the encryption process, whereas the NIAR scheme requires $17nb$ group multiplications.
Regarding the routing computation time, the primary computational cost in the decryption phase arises from bilinear pairing operations. We therefore compare the number of bilinear pairings required by each scheme. The NIAR scheme incurs $64n^2b$ pairing operations, while our protocol requires only $n^2b/l$ pairings, representing a significant reduction in computational cost, especially as $l$ increases.

\subsubsection{Transmission Efficiency}
According to equation (\ref{equ_ct}), each encryption operation in our protocol generates a ciphertext represented as a $1\times 8$ vector in $\mathbb{G}_1$.
Since each element in $\mathbb{G}_1$ requires $49$ bytes of storage,
the total ciphertext size per user amounts to $392$ bytes.
For a plaintext message of length $l$ bytes per user, the resulting plaintext-to-ciphertext expansion factor is calculated as $392/l$.

In contrast, the NIAR construction in~\cite{ShiWu:C21} encrypts each bit into a single $\texttt{MCFE}^{ffh}$ ciphertext, as outlined in the $\texttt{MCFE}^{ffh}$ framework.
An $\texttt{MCFE}^{ffh}$ ciphertext consists of an $\texttt{MCFE}$ ciphertext combined with the output of $\texttt{FE.KGEN}$.
According to the $\texttt{MCFE}$ construction, an $\texttt{MCFE}$ ciphertext is composed of $\left[\!\left[c_1\right]\!\right]$, $\left[\!\left[c_2\right]\!\right]$ and $\left[\!\left[\tilde{c}\right]\!\right]$, with respective lengths of $2b$, $2$, and $2$ group elements, resulting in a total length of $2b + 4$.
Here, $b=1$ when encrypting bit-by-bit.
Additionally, the output of $\texttt{FE.KGEN}$ has a length of $2$ group elements, resulting in a total $\texttt{MCFE}^{ffh}$ ciphertext length of $2b + 6$. For a single plaintext bit ($b=1$), the resulting ciphertext thus contains $8$ group elements in $\mathbb{G}_1$.
Using the provided test curve parameters, where each $\mathbb{G}_1$ element occupies $392$ bytes, the ciphertext size is calculated to be $3136$ bytes, yielding an expansion ratio of $3136$ for a $1$-bit plaintext.
This is markedly higher than the expansion factor of our proposed protocol.

In summary, our protocol demonstrates significantly greater communication efficiency than the NIAR scheme, both in terms of ciphertext expansion and computational cost, offering substantial improvements for practical applications.

\section{Conclusion}\label{sec_conclu}
Private and covert communication is crucial in meeting the ever-evolving security demands of the interconnected digital era.
Anonymous routing technology emerges as a pivotal tool, aiming to obfuscate communication paths and participant identities, and resist sophisticated traffic analysis attacks.
This paper proposes a decentralized anonymous communication network protocol that fundamentally eliminates dependence on threshold trust models and trusted setups, which are common vulnerabilities in existing systems.
The protocol's resistance to analysis and traceability is validated through formal security proofs, which remains inherently robust in the face of evolving adversarial strategies.
Furthermore, simulation results demonstrate that the proposed protocol is also efficient in practical deployment.

Several issues remain to be addressed in future work.
First, the proposed protocol addresses the identity protection of trusted senders. In future work, we will further explore integrating Private Information Retrieval methods to achieve anonymous protection for receivers, thereby realizing bidirectional identity untraceability.
Second, the protocol's applicable scale is currently limited. However, its decentralized nature allows for integration with an efficient grouping strategy, thereby further optimizing the computational efficiency of routing computation and enhancing the protocol’s dynamic adaptability.



{\appendices
\section{Proof of Proposition \ref{prop_1}}
Proposition \ref{prop_1} can be considered as a vector version of the XDLin assumption.
We begin by proving the case for $m=2$ via the contradiction method.

For $m=2$, we have
\begin{equation*}
P_0=
\bigg(\left[\!\left[a\right]\!\right]_{1,2},
       \left[\!\left[b\right]\!\right]_{1,2},
     [\![\begin{pmatrix}
       a k_1\\a k_2\end{pmatrix}]\!]_{1,2},
      [\![\begin{pmatrix}
       b d_1\\b d_2\end{pmatrix}]\!]_{1,2},
      [\![\begin{pmatrix}
       k_1+d_1\\k_2+d_2 \end{pmatrix}]\!]_{x} \bigg), a,b,k_1,k_2,d_1,d_2 \overset{{\scriptscriptstyle\$}}{\leftarrow} \mathbb{Z}_q;
\end{equation*}
\begin{equation*}
P_1=
\bigg([\![a]\!]_{1,2},
    [\![b]\!]_{1,2},
    [\![\begin{pmatrix}
       a k_1\\a k_2\end{pmatrix}]\!]_{1,2},
    [\![\begin{pmatrix}
       b d_1\\b d_2\end{pmatrix}]\!]_{1,2},
    [\![\begin{pmatrix}
       e_1\\e_2 \end{pmatrix}]\!]_{x} \bigg),
       a,b,k_1,k_2,d_1,d_2,e_1,e_2 \overset{{\scriptscriptstyle\$}}{\leftarrow} \mathbb{Z}_q.
\end{equation*}
Let
\begin{equation*}
P'_1:= 
\bigg(\left[\!\left[a\right]\!\right]_{1,2},
       \left[\!\left[b\right]\!\right]_{1,2},
     [\![\begin{pmatrix}
       a k_1\\a k_2\end{pmatrix}]\!]_{1,2},
      [\![\begin{pmatrix}
       b d_1\\b d_2\end{pmatrix}]\!]_{1,2},
      [\![\begin{pmatrix}
       k_1+d_1\\e_1 \end{pmatrix}]\!]_{x} \bigg),
       a,b,k_1,k_2,d_1,e_1 \overset{{\scriptscriptstyle\$}}{\leftarrow} \mathbb{Z}_q.
\end{equation*}
With the XDLin assumption, the indistinguishability of $P_0$ and $P_1$ is equivalent to the indistinguishability of $P_1$ and $P'_1$.
Suppose there exists a p.p.t. $\mathcal{A}$ such that the advantage $\left| \mathrm{Pr}\left[ 1 \leftarrow \mathcal{A}(P_1) \right] - \mathrm{Pr}\left[ 1 \leftarrow \mathcal{A}(P'_1) \right] \right|$ is non-negligible. This would imply that $\mathcal{A}$ could distinguish the two distributions defined in Definition 1, violating the XDLin assumption. Therefore, we can conclude that $P_1$ and $P'_1$ are computationally indistinguishable, establishing Proposition \ref{prop_1} for the case $m=2$.
The argument generalizes to arbitrary $m$ by inductively applying the same reasoning to each additional dimension.
Thus, we can conclude that Proposition \ref{prop_1} holds.

\section{Proof of Proposition \ref{prop_2}}
Given that
$
\left([\![a]\!]_{1,2},[\![b]\!]_{1,2},[\![a\boldsymbol{l}]\!]_{1,2},[\![b\boldsymbol{e}]\!]_{1,2},\boldsymbol{y}\right),
$
where $a,b\overset{{\scriptscriptstyle\$}}{\leftarrow} \mathbb{Z}_q$, $\boldsymbol{l},\boldsymbol{e}\overset{{\scriptscriptstyle\$}}{\leftarrow} \mathbb{Z}_q^{(m-1)\times 1}$, and $\boldsymbol{y}= [\![\boldsymbol{l}+\boldsymbol{e}]\!]_{x}$ or  $ [\![\boldsymbol{l}'+\boldsymbol{e}]\!]_{x}$, $\boldsymbol{l}' \overset{{\scriptscriptstyle\$}}{\leftarrow} \mathbb{Z}_q^{(m-1)\times 1}$.
Denote the $i$-th element of $\boldsymbol{l}$, $\boldsymbol{e}$, $\boldsymbol{y}$ by $l_i$, $e_i$, $y_i$, respectively. Let
\begin{equation*}
Q:=
\Big([\![a]\!]_{1,2},[\![b]\!]_{1,2},
[\![a\boldsymbol{l} \| a(t-\sum_{i=1}^{m-1}l_i) ]\!]_{1,2},
[\![b\boldsymbol{e}  \| b(r- \sum_{i=1}^{m-1}e_i)]\!]_{1,2},
\boldsymbol{y} \| (t+r-\sum_{i=1}^{m-1}y_i)
 \Big),
r\overset{{\scriptscriptstyle\$}}{\leftarrow} \mathbb{Z}_q.
\end{equation*}
We note that when $\boldsymbol{y}= [\![\boldsymbol{l}+\boldsymbol{e}]\!]_{x}$, $Q$ is equivalent to $P_0$ and when $\boldsymbol{y}= [\![\boldsymbol{l}'+\boldsymbol{e}]\!]_{x}$, $Q$ is equivalent to $P_1$.
Meanwhile, condition $\sum_{i=1}^{m}k_i = t = \sum_{i=1}^{m}k'_i$ in Proposition \ref{prop_2} is satisfied.
Therefore, by Proposition \ref{prop_1}, Proposition \ref{prop_2} can be directly derived.

\section{Proof of Proposition \ref{prop_3}}
The indistinguishability of the target games in Proposition \ref{prop_3} is formally established via a security reduction to the game variant $(\bigtriangledown)$, defined as:
     \begin{itemize}
    \item $\mathcal{A}$ receives
        \begin{itemize}
        \item $P_0^{(2)} := ([\![a]\!]_{1,2},[\![b]\!]_{1,2},[\![a\boldsymbol{k}]\!]_{1,2},[\![b\boldsymbol{d}]\!]_{1,2},[\![\boldsymbol{k}+\boldsymbol{d}]\!]_{x})$, where $a,b\overset{{\scriptscriptstyle\$}}{\leftarrow} \mathbb{Z}_q$ and $\boldsymbol{k},\boldsymbol{d}\overset{{\scriptscriptstyle\$}}{\leftarrow} \mathbb{Z}_q^{m\times 1}$,
        \item $P_1^{(2)} := ([\![a]\!]_{1,2},[\![b]\!]_{1,2},[\![a\boldsymbol{k}]\!]_{1,2},[\![b\boldsymbol{d}]\!]_{1,2},[\![\boldsymbol{k}'+\boldsymbol{d}]\!]_{x})$, where $a,b\overset{{\scriptscriptstyle\$}}{\leftarrow} \mathbb{Z}_q$ and $\boldsymbol{k},\boldsymbol{k}',\boldsymbol{d}\overset{{\scriptscriptstyle\$}}{\leftarrow} \mathbb{Z}_q^{m\times 1}$,
        \end{itemize}
        satisfying  $\sum_{i=1}^{m}k_i = \sum_{i=1}^{m}k'_i=t$, where $k_i, k'_i$ are the $i$-th elements of $\boldsymbol{k}$ and $\boldsymbol{k}'$, respectively.
    \item $\mathcal{A}$ uniformly samples elements $\{z_1,\cdots,z_m\}$ from $\mathbb{Z}_q$, and sends them to $\mathcal{C}$. Subsequently, $\mathcal{C}$ returns $\{[\![a z_1]\!]_{x},\cdots,[\![a z_m]\!]_{x}\}$ to $\mathcal{A}$.
    \end{itemize}

    We note that $\mathcal{A}$ gains no computational advantage in extracting the secret exponent $a$, since any efficient extraction of $a$ would imply breaking the discrete logarithm problem. Combining with Proposition~\ref{prop_2}, we can conclude that for any p.p.t. $\mathcal{A}$, the two distributions $P_0^{(2)}$ and $P_1^{(2)}$ in game variant $(\bigtriangledown)$ are computationally indistinguishable.
    Meanwhile, game $(\triangle)$ exhibits strictly greater computational hardness than $(\bigtriangledown)$, as the adversary has access to more information in $(\bigtriangledown)$.
    Therefore, $P_0^{(2)}$ and $P_1^{(2)}$ remain computationally indistinguishable in game $(\triangle)$, which confirms that Proposition~\ref{prop_3} holds.

\section{Proof of Proposition \ref{prop_4}}
First, we observe that the adversarial interaction where $\mathcal{A}$ uniformly samples elements $\{r_1,\cdots,r_m\}$ from the group $\mathbb{G}_x$ and submits them to $\mathcal{C}$, who then returns $\{r_1^a,\cdots,r_m^a\}$, is computationally indistinguishable from the scenario where $\mathcal{C}$ directly transmits the pairs $\left\{ (r_1,r_1^{a}),\cdots,(r_m,r_m^{a})\right\}$ to $\mathcal{A}$.
The equivalence comes from the fact that the information ultimately obtained by $\mathcal{A}$ in both interactions is identically distributed, provided that the selection of $\{r_1,\cdots,r_m\}$ is random and cannot be controlled by $\mathcal{A}$.
Additionally, based on the discrete logarithm problem, transmitting the pairs $\left\{ (r_1,r_1^a),\cdots,(r_m,r_m^a)\right\}$ to $\mathcal{A}$ is equivalent to transmitting the pairs $\big\{ (r_1,r_1^{\frac{1}{a}}),\cdots,(r_m,r_m^{\frac{1}{a}})\big\}$.
The equivalence lies in that $\mathcal{A}$ cannot efficiently deduce $a$ or $a^{-1}$ from the transmitted elements in both cases.
Therefore, Proposition \ref{prop_4} holds.

\section{Proof of Lemma \ref{lem_hyb1}}
To prove  Lemma \ref{lem_hyb1}, it suffices to show that for any $\ell < Q$,
\begin{equation}\label{equ_tilde}
    \left| \mathrm{Pr}\left[ 1 \leftarrow \mathcal{A}\left(\textbf{Hyb}^{(1,\ell)}\right) \right] - \mathrm{Pr}\left[ 1 \leftarrow \mathcal{A}\left(\textbf{Hyb}^{(1,\ell+1)}\right) \right] \right| 
     \leq 2\left[ \frac{m}{q}+(m-1) \mathrm{Adv}^{\mathrm{XDLin}}(\lambda)  \right].
\end{equation}
To this end, we consider the following hybrid experiments to establish a switch between $\textbf{Hyb}^{(1,\ell)}$ and $\textbf{Hyb}^{(1,\ell+1)}$.

\subsubsection{Hybrid Experiment $\widetilde{\textbf{Hyb}}^{(1,\ell)}$}\label{sssec_hyb_1.1}
$\widetilde{\textbf{Hyb}}^{(1,\ell)}$ is defined as a modification of $\textbf{Hyb}^{(1,\ell)}$, and the only difference lies in the calculation of $\{\boldsymbol{ct}_{1,\ell+1},\cdots,\boldsymbol{ct}_{m,\ell+1}\}$.
In $\widetilde{\textbf{Hyb}}^{(1,\ell)}$, the third-to-last element $0$ is replaced by a randomly selected element $\zeta_{i} \overset{{\scriptscriptstyle\$}}{\leftarrow} \mathbb{Z}_q^{\times}$. Specifically, $\mathcal{C}$ computes $\boldsymbol{ct}_{i,\ell+1}$ by
$$
\boldsymbol{ct}_{i,\ell+1} :=[\![C_i \left(K_{i,\ell+1},\alpha_{i,\ell+1},\alpha'_{i,\ell+1},0,0,\zeta_i,x_{i,\ell+1},0
    \right)^{\top} ]\!]_{1}.
$$
We can prove the computational indistinguishability of $\widetilde{\textbf{Hyb}}^{(1,\ell)}$ and $\textbf{Hyb}^{(1,\ell)}$ through the reduction to the XDLin assumption.
Here we demonstrate how to construct a non-uniform p.p.t. adversary $\mathcal{B}$ against the XDLin problem by leveraging the adversary $\mathcal{A}$.
\begin{itemize}
    \item After $\mathcal{B}$ receives
    \begin{equation*}
        [\![ \begin{pmatrix}
        a_1\\\vdots\\a_m
        \end{pmatrix} ]\!]_{1,2},
        [\![ \begin{pmatrix}
        b_1\\\vdots\\b_m
        \end{pmatrix} ]\!]_{1,2},
        [\![ \begin{pmatrix}
        a_1\alpha_{1,\ell+1} \\\vdots\\a_m\alpha_{m,\ell+1}
        \end{pmatrix} ]\!]_{1,2}, 
        [\![ \begin{pmatrix}
        b_1\alpha'_{1,\ell+1} \\\vdots\\b_m\alpha'_{m,\ell+1}
        \end{pmatrix} ]\!]_{1,2},
        [\![ \begin{pmatrix}
        Y_1 \\\vdots\\Y_m
        \end{pmatrix} ]\!]_{1},
    \end{equation*}
    it needs to distinguish whether
    $Y_i = \alpha_{i,\ell+1} + \alpha'_{i,\ell+1}$
    or
    $Y_i = \alpha_{i,\ell+1} + \alpha'_{i,\ell+1}+\zeta_i$.
    \item $\mathcal{B}$ calls $\mathtt{Setup}(1^{\lambda},n)$ to obtain $\boldsymbol{tk}$.
    \item For each $i\in [n]$, $\mathcal{C}$ samples an invertible matrix $W_i\overset{{\scriptscriptstyle\$}}{\leftarrow} (\mathbb{Z}_q)_{8\times 8}^{\times}$. Let
    $$
     C_i =
     W_i \begin{pmatrix}
    1 & & & & & & &  \\
      &a_i & & & & & & \\
      & & b_i & & & & & \\
      & & & a_i & & & & \\
      & & & & a_i & & & \\
    0 & 1 & 1 & 0 & 0 & 1 & 0 & 0 \\
      & & & & & & a_i& \\
      & & & & & & & a_i
    \end{pmatrix},
    $$
    and
    \begin{multline*}
    R_i  =  \mathrm{diag}\left(1,\frac{1}{a_ib_i},\cdots,\frac{1}{a_ib_i} \right)
    \begin{pmatrix}
    1 & & & & & & &  \\
      &b_i & & & & & & \\
      & & a_i & & & & & \\
      & & & b_i & & & & \\
      & & & & b_i & & & \\
    0 & -b_i & -a_i & 0 & 0 & a_ib_i & 0 & 0 \\
      & & & & & & b_i& \\
      & & & & & & & b_i
    \end{pmatrix} W_i^{-1}.
    \end{multline*}
    We have $R_iC_i=\mathbf{I}_{8}$.

    We remark that the $j$-th token corresponding to sender $i$ is computed by
    \begin{align*}
        \boldsymbol{tk}_{i}^j &:= [\![\left(\rho_j,0,0,a_ib_i\beta_{i,j},a_ib_i\gamma_{i,j},0,  a_ib_i\delta_{j,\boldsymbol{\pi}(i)},0\right)R_i]\!]_{2}\\
        &=[\![\left(\rho_j,0,0,\beta_{i,j},\gamma_{i,j},0,  \delta_{j,\boldsymbol{\pi}(i)},0\right)\begin{pmatrix}
    1 & & & & & & &  \\
      &b_i & & & & & & \\
      & & a_i & & & & & \\
      & & & b_i & & & & \\
      & & & & b_i & & & \\
    0 & -b_i & -a_i & 0 & 0 & a_ib_i & 0 & 0 \\
      & & & & & & b_i& \\
      & & & & & & & b_i
    \end{pmatrix} W_i^{-1}]\!]_{2}
    \end{align*}
    Since $\beta_{ij},\gamma_{ij} \overset{{\scriptscriptstyle\$}}{\leftarrow} \mathbb{Z}_q$, we have $a_ib_i\beta_{ij}$ and $ a_ib_i\gamma_{ij}$ are still random sampled elements.
    Here we note that $\theta_i = a_ib_i$.

    \item $\mathcal{A}$ makes queries $\{x_{i,t}\}_{i\leq m}$.
    When $t=\ell+1$, $\boldsymbol{ct}_{i,\ell+1}$ is computed by
    \begin{equation*}
        \boldsymbol{ct}_{i,\ell+1} := [\![W_i \big(
    K_{i,\ell+1}, a_i\alpha_{i,\ell+1},
    b_i\alpha'_{i,\ell+1}, 
    0,0,
    Y_i, a_i x_{i,\ell+1},0
    \big)^{\top} ]\!]_{1}.
    \end{equation*}

    Otherwise, $\mathcal{B}$ computes the ciphertexts in the same way as in $\textbf{Hyb}^{(1,\ell)}$, and returns them to $\mathcal{A}$.
    \item $\mathcal{B}$ outputs the same guess that $\mathcal{A}$ outputs.
\end{itemize}
Observe that if $Y_i = \alpha_{i,\ell+1} + \alpha'_{i,\ell+1}$, then $\mathcal{A}$'s view is identically distributed as in $\textbf{Hyb}^{(1,\ell)}$; otherwise, if $Y_i = \alpha_{i,\ell+1} + \alpha'_{i,\ell+1}+\zeta_i$, $\mathcal{A}$'s view is identically distributed as in $\widetilde{\textbf{Hyb}}^{(1,\ell)}$.
By Proposition \ref{prop_2}, we can obtain
\begin{equation*}
    \left| \mathrm{Pr}\left[ 1 \leftarrow \mathcal{A}\left(\textbf{Hyb}^{(1,\ell)}\right) \right] - \mathrm{Pr}\left[ 1 \leftarrow \mathcal{A}\left(\widetilde{\textbf{Hyb}}^{(1,\ell)}\right) \right] \right| 
     \leq  \frac{m}{q}+(m-1) \mathrm{Adv}^{\mathrm{XDLin}}(\lambda) .
\end{equation*}
Therefore, the experiments $\textbf{Hyb}^{(1,\ell)}$ and $\widetilde{\textbf{Hyb}}^{(1,\ell)}$ are computationally indistinguishable.
Next, we aim to demonstrate that $\widetilde{\textbf{Hyb}}^{(1,\ell)}$ is equivalent to the following experiment $(*)$.

\subsubsection{Hybrid Experiment $(*)$}\label{sssec_hyb_1.2}
The difference between games $\widetilde{\textbf{Hyb}}^{(1,\ell)}$ and $(*)$ lies in the calculation of $\boldsymbol{ct}_{i,\ell+1}$.
In $(*)$,
\begin{equation*}
    \boldsymbol{ct}_{i,\ell+1} :=
    [\![C_i \big(K_{i,\ell+1},\alpha_{i,\ell+1},\alpha'_{i,\ell+1},0,0,
    \zeta_i,x_{i,\ell+1},x_{\boldsymbol{\pi}^{-1}(i),\ell+1}
    \big)^{\top} ]\!]_{1}.
\end{equation*}
By computing the product of
$$
C_i
\begin{pmatrix}
 \boldsymbol{I}_{5}  &\\
 & \begin{pmatrix}
1 & 0 & 0 \\
0 & 1 & 0 \\
\frac{x_{\boldsymbol{\pi}^{-1}(i),\ell + 1}}{\zeta_i} & 0 & 1
\end{pmatrix}
\end{pmatrix}
$$
and
$\big(K_{i,\ell+1},\alpha_{i,\ell+1},\alpha'_{i,\ell+1},0,0,\zeta_i,x_{i,\ell+1},0\big)^{\top}$,
we observe that the ciphertext in the exponent $g_1$ remains
$$
C_i\left(
K_{i,\ell+1},\alpha_{i,\ell+1},\alpha'_{i,\ell+1},0,0,\zeta_i,x_{i,\ell+1},x_{\boldsymbol{\pi}^{-1}(i),\ell+1} \right)^{\top}.
$$
Moreover, the product 
$$
\begin{pmatrix}
 \boldsymbol{I}_{5}  &\\
 & \begin{pmatrix}
1 & 0 & 0 \\
0 & 1 & 0 \\
-\frac{x_{\boldsymbol{\pi}^{-1}(i),\ell + 1}}{\zeta_i} & 0 & 1
\end{pmatrix}
\end{pmatrix}R_iC_i
\begin{pmatrix}
 \boldsymbol{I}_{5}  &\\
 & \begin{pmatrix}
1 & 0 & 0 \\
0 & 1 & 0 \\
\frac{x_{\boldsymbol{\pi}^{-1}(i),\ell + 1}}{\zeta_i} & 0 & 1
\end{pmatrix}
\end{pmatrix}=\mathbf{I}_{8}.
$$
Hence, we conclude that $\widetilde{\textbf{Hyb}}^{(1,\ell)}$ and $(*)$ are equivalent. Formally,
$$
 \mathrm{Pr}\left[ 1 \leftarrow \mathcal{A}\left(\widetilde{\textbf{Hyb}}^{(1,\ell)}\right) \right]  = \mathrm{Pr}\left[ 1 \leftarrow \mathcal{A}\left(*\right) \right].
$$

Symmetrically to the first step, we can eliminate $\zeta_i$ and obtain that
\begin{equation*}
    \left|\mathrm{Pr}\left[ 1 \leftarrow \mathcal{A}\left(*\right) \right] - \mathrm{Pr}\left[ 1 \leftarrow \mathcal{A}\left(\textbf{Hyb}^{(1,\ell+1)}\right) \right] \right|
     \leq  \frac{m}{q}+(m-1) \mathrm{Adv}^{\mathrm{XDLin}}(\lambda).
\end{equation*}
Therefore, equation (\ref{equ_tilde}) in Lemma \ref{lem_hyb1} holds, demonstrating that $\textbf{Hyb}^{(1,Q)}$ and $\textbf{Hyb}^{(1,0)}$ are computationally indistinguishable.

\section{Proof of Lemma \ref{lem_hyb2}}
To prove Lemma \ref{lem_hyb2}, it suffices to show that
\begin{equation}\label{equ_hyb2l}
    \left| \mathrm{Pr}\left[ 1 \leftarrow \mathcal{A}\left(\textbf{Hyb}^{(2,\ell)}\right) \right] - \mathrm{Pr}\left[ 1 \leftarrow \mathcal{A}\left(\textbf{Hyb}^{(2,\ell+1)}\right) \right] \right| 
     \leq 2\left[ \frac{m+1}{q}+(Qm+m-Q) \mathrm{Adv}^{\mathrm{XDLin}}(\lambda)  \right].
\end{equation}
To this end, we consider the following hybrid experiments $\textbf{Hyb}^{(2,\ell,1)}$, $\textbf{Hyb}^{(2,\ell,2)}$, and $\textbf{Hyb}^{(2,\ell,3)}$ to establish a switch between $\textbf{Hyb}^{(2,\ell)}$ and $\textbf{Hyb}^{(2,\ell+1)}$.

\subsubsection{Hybrid Experiment $\textbf{Hyb}^{(2,\ell,1)}$}\label{sssec_hyb_2.1}
The difference between $\textbf{Hyb}^{(2,\ell,1)}$ and $\textbf{Hyb}^{(2,\ell)}$ lies in the $(\ell+1)$-th row of token $\boldsymbol{tk}_{i}$.
In $\textbf{Hyb}^{(2,\ell,1)}$,
$$
\boldsymbol{tk}_{i}^{\ell+1}:= [\![\left(\rho_{\ell+1},0,0,\beta_{i,\ell+1},\gamma_{i,\ell+1},\xi_i,\delta_{\ell+1,\boldsymbol{\pi}(i)},0\right)R_i]\!]_{2},
$$
where $\xi_i \overset{{\scriptscriptstyle\$}}{\leftarrow} \mathbb{Z}_q^{\times}$.
Similar to the proof of Lemma \ref{lem_hyb1},
leveraging the elements $\beta_{i,\ell+1}$ and $\gamma_{i,\ell+1}$, we can construct a non-uniform p.p.t. adversary $\mathcal{B}$ against the XDLin problem by leveraging the adversary $\mathcal{A}$.
\begin{itemize}
    \item After $\mathcal{B}$ receives
    \begin{equation*}
        [\![ \begin{pmatrix}
        a_1\\\vdots\\a_m
        \end{pmatrix} ]\!]_{1,2},
        [\![ \begin{pmatrix}
        b_1\\\vdots\\b_m
        \end{pmatrix} ]\!]_{1,2},
        [\![ \begin{pmatrix}
        a_1\beta_{1,\ell+1} \\\vdots\\a_m\beta_{m,\ell+1}
        \end{pmatrix} ]\!]_{1,2}, 
        [\![ \begin{pmatrix}
        b_1\gamma_{1,\ell+1} \\\vdots\\b_m\gamma_{m,\ell+1}
        \end{pmatrix} ]\!]_{1,2},
        [\![ \begin{pmatrix}
        Y_1 \\\vdots\\Y_m
        \end{pmatrix} ]\!]_{2},
    \end{equation*}
    it needs to distinguish whether
    $Y_i = \beta_{i,\ell+1} + \gamma_{i,\ell+1}$
    or
    $Y_i = \beta_{i,\ell+1} + \gamma_{i,\ell+1}+\xi_i$.
    \item $\mathcal{B}$ calls $\mathtt{Setup}(1^{\lambda},n)$ to obtain $\boldsymbol{tk}$.
    \item For each $i\in [n]$, $\mathcal{C}$ samples an invertible matrix $W_i\overset{{\scriptscriptstyle\$}}{\leftarrow} (\mathbb{Z}_q)_{8\times 8}^{\times}$. Let
    $$
     C_i =
     W_i \begin{pmatrix}
    1 & & & & & & &  \\
      &b_i & & & & & & \\
      & & b_i & & & & & \\
      & & & b_i & &-b_i & & \\
      & & & & a_i &-a_i & & \\
      & & & & & a_ib_i & & \\
      & & & & & & b_i& \\
      & & & & & & & b_i
    \end{pmatrix},
    $$
    and
    \begin{equation*}
    R_i  =  \mathrm{diag}\left(1,\frac{1}{a_ib_i},\cdots,\frac{1}{a_ib_i} \right)
    \begin{pmatrix}
    1 & & & & & & &  \\
      &a_i & & & & & & \\
      & & a_i & & & & & \\
      & & & a_i & &1 & & \\
      & & & & b_i &1 & & \\
      & & & & & 1 &  &  \\
      & & & & & & a_i& \\
      & & & & & & & a_i
    \end{pmatrix} W_i^{-1}.
    \end{equation*}
    We have $R_iC_i = \mathbf{I}_{8}$.

    When $j\neq \ell+1$, the $j$-th token corresponding to sender $i$ is computed by
    \begin{align*}
        \boldsymbol{tk}_{i}^j &:= [\![\left(\rho_j,0,0,a_ib_i\beta_{i,j},a_ib_i\gamma_{i,j},0,  a_ib_i\delta_{j,\boldsymbol{\pi}(i)},0\right)R_i]\!]_{2}\\
        &=[\![\left(\rho_j,0,0,\beta_{i,j},\gamma_{i,j},0,  \delta_{j,\boldsymbol{\pi}(i)},0\right)\begin{pmatrix}
    1 & & & & & & &  \\
      &a_i & & & & & & \\
      & & a_i & & & & & \\
      & & & a_i & &1 & & \\
      & & & & b_i &1 & & \\
      & & & & & 1 &  &  \\
      & & & & & & a_i& \\
      & & & & & & & a_i
    \end{pmatrix} W_i^{-1}]\!]_{2}
    \end{align*}
    Here we note that $a_ib_i\beta_{i,j},~a_ib_i\gamma_{i,j}$ remain random samples and $\theta_i =a_ib_i$.
    When $j = \ell+1$,
    \begin{equation*}
       \boldsymbol{tk}_{i}^{\ell+1}:= [\![\big(\rho_{\ell+1},0,0,a_i\beta_{i,\ell+1},b_i\gamma_{i,\ell+1},
       Y_i,  a_i\delta_{\ell+1,\boldsymbol{\pi}(i)},0\big)W_i^{-1}]\!]_{2}
    \end{equation*}

    \item $\mathcal{B}$ computes the ciphertexts in the same way as in $\textbf{Hyb}^{(2,\ell)}$, and returns them to $\mathcal{A}$.
    \item $\mathcal{B}$ outputs the same guess that $\mathcal{A}$ outputs.
\end{itemize}
Observe that if $Y_i = \beta_{i,\ell+1} + \gamma_{i,\ell+1}$, then $\mathcal{A}$'s view is identically distributed as in $\textbf{Hyb}^{(2,\ell)}$; otherwise, if $Y_i = \beta_{i,\ell+1} + \gamma_{i,\ell+1}+\xi_i$, $\mathcal{A}$'s view is identically distributed as in $\textbf{Hyb}^{(2,\ell,1)}$.
Based on Proposition \ref{prop_1} and the XDLin assumption, we can obtain that
\begin{equation*}
    \Big| \mathrm{Pr}\left[ 1 \leftarrow \mathcal{A}\left(\textbf{Hyb}^{(2,\ell)}\right) \right] - \mathrm{Pr}\left[ 1 \leftarrow \mathcal{A}\left(\textbf{Hyb}^{(2,\ell,1)}\right) \right] \Big| 
     \leq  \frac{m}{q}+m \mathrm{Adv}^{\mathrm{XDLin}}(\lambda).
\end{equation*}
Next, we demonstrate that $\textbf{Hyb}^{(2,\ell,1)}$ is equivalent to the following $\textbf{Hyb}^{(2,\ell,2)}$.

\subsubsection{Hybrid Experiment $\textbf{Hyb}^{(2,\ell,2)}$}\label{sssec_hyb_2.2}
In $\textbf{Hyb}^{(2,\ell,2)}$, the $(\ell+1)$-th row of token $\boldsymbol{tk}_{i}$ is replaced by
$$
\boldsymbol{tk}_{i}^{\ell+1}:= [\![\left(0,0,0,0,0,\xi_i,0,0\right)R_i]\!]_{2}.
$$
When computing $\boldsymbol{ct}_{i,t}$, let
    \begin{equation*}
        \boldsymbol{ct}_{i,t} :=[\![C_i \big(
    K_{i,t},\alpha_{i,t},\alpha'_{i,t},0,0,
    \frac{1}{\xi_i}(\rho_{\ell+1}K_{i,t}+\theta_i\delta_{\ell+1,\boldsymbol{\pi}(i)}x_{i,t})
    ,x_{i,t},x_{\boldsymbol{\pi}^{-1}(i),t}
    \big)^{\top} ]\!]_{1}.
    \end{equation*}
It can be validated that the product of tokens and ciphertexts remains the same.
Moreover, the product of
$$
\begin{pmatrix}
    1 & & & & & & &  \\
      &1 & & & & & & \\
      & & 1 & & & & & \\
      & & & 1 & & & & \\
      & & & & 1 & & & \\
    \frac{-\rho_{\ell+1}}{\xi_i} & 0 & 0 & \frac{-\beta_{i,\ell+1}}{\xi_i} & \frac{-\gamma_{i,\ell+1}}{\xi_i} & 1 & \frac{-\theta_i\delta_{\ell+1,\boldsymbol{\pi}(i)}}{\xi_i} & 0 \\
      & & & & & & 1& \\
      & & & & & & & 1
    \end{pmatrix}R_i
$$
and
$$
    C_i\begin{pmatrix}
    1 & & & & & & &  \\
      &1 & & & & & & \\
      & & 1 & & & & & \\
      & & & 1 & & & & \\
      & & & & 1 & & & \\
    \frac{\rho_{\ell+1}}{\xi_i} & 0 & 0 & \frac{\beta_{i,\ell+1}}{\xi_i} & \frac{\gamma_{i,\ell+1}}{\xi_i} & 1 & \frac{\theta_i\delta_{\ell+1,\boldsymbol{\pi}(i)}}{\xi_i} & 0 \\
      & & & & & & 1& \\
      & & & & & & & 1
    \end{pmatrix}
$$
equals $\mathbf{I}_8$.
Hence, we can obtain that $\textbf{Hyb}^{(2,\ell,1)}$ is equivalent to $\textbf{Hyb}^{(2,\ell,2)}$. Formally,
$$
 \mathrm{Pr}\left[ 1 \leftarrow \mathcal{A}\left(\textbf{Hyb}^{(2,\ell,1)}\right) \right]  = \mathrm{Pr}\left[ 1 \leftarrow \mathcal{A}\left(\textbf{Hyb}^{(2,\ell,2)}\right) \right].
$$
Following this, we demonstrate that $\textbf{Hyb}^{(2,\ell,2)}$ is computationally indistinguishable from the following $\textbf{Hyb}^{(2,\ell,3)}$.

\subsubsection{Hybrid Experiment $\textbf{Hyb}^{(2,\ell,3)}$}\label{sssec_hyb_2.3}
$\textbf{Hyb}^{(2,\ell,3)}$ differs from $\textbf{Hyb}^{(2,\ell,2)}$ in the calculation of $\{\boldsymbol{ct}_{i,t}\}_{i\leq m}$, where $K_{i,t}$ in the third-to-last element is replaced by a randomly selected element $K'_{i,t} \overset{{\scriptscriptstyle\$}}{\leftarrow} \mathbb{Z}_q$ satisfying
$\sum_{i = 1}^{m}K'_{i,t}=\sum_{i = 1}^{m}K_{i,t}.$
In $\textbf{Hyb}^{(2,\ell,3)}$, the challenger computes $\boldsymbol{ct}_{i,t}$ by
\begin{equation*}
    \boldsymbol{ct}_{i,t} :=[\![C_i \big(
K_{i,t},\alpha_{i,t},\alpha'_{i,t},0,0,
\frac{1}{\xi_i}(\rho_{\ell+1}K'_{i,t}+\theta_i\delta_{\ell+1,\boldsymbol{\pi}(i)}x_{i,t}),x_{i,t},x_{\boldsymbol{\pi}^{-1}(i),t}
\big)^{\top} ]\!]_{1}.
\end{equation*}

Following the methodology in Appendix E, the indistinguishability of $\textbf{Hyb}^{(2,\ell,2)}$ and $\textbf{Hyb}^{(2,\ell,3)}$ can be established through a reduction to the XDLin assumption.
Here we demonstrate how to construct a non-uniform p.p.t. adversary $\mathcal{B}$ against game $(\diamond)$ by leveraging adversary $\mathcal{A}$.
\begin{itemize}
    \item
    $\mathcal{B}$ interacts with $\mathcal{C}$ to obtain public parameters.
    Then, $\mathcal{C}$ generates and returns
    \begin{multline*}
        [\![ \rho_{\ell+1} ]\!]_{1,2},
        [\![ u ]\!]_{1,2},
        [\![ \begin{pmatrix}
        \rho_{\ell+1} K_{1,1} & \cdots & \rho_{\ell+1} K_{1,Q}\\
        \vdots  & \ddots & \vdots\\
        \rho_{\ell+1} K_{m,1} & \cdots & \rho_{\ell+1} K_{m,Q}
        \end{pmatrix} ]\!]_{1,2},\\
        [\![ \begin{pmatrix}
        \alpha_{1,1} & \cdots & \alpha_{1,Q}\\
        \vdots  & \ddots & \vdots\\
        \alpha_{m,1} & \cdots & \alpha_{m,Q}
        \end{pmatrix} ]\!]_{1,2},
        [\![ \begin{pmatrix}
        Y_{1,1} & \cdots & Y_{1,Q}\\
        \vdots  & \ddots & \vdots\\
        Y_{m,1} & \cdots & Y_{m,Q}
        \end{pmatrix} ]\!]_{1}.
    \end{multline*}
    Meanwhile, $\mathcal{C}$ sends $\{[\![ \frac{\rho_{j}}{u} ]\!]_{2}\}_{j\in [n]}$ to $\mathcal{B}$.
    Here we remark that, by our protocol, the generation of $\{\rho_j\}_{j\in [n]}$ cannot be controlled by the adversary and the values of $\{\rho_j\}_{j\in [n]}$ are secret.

    \item After $\mathcal{B}$ receives the above instances, it needs to distinguish whether
    $Y_{i,j} = K_{i,j} + \frac{\alpha_{i,j}}{u} $
    or
    $Y_{i,j} = K'_{i,j} + \frac{\alpha_{i,j}}{u}$, where $\sum_{i=1}^{m} K_{i,t} = \sum_{i=1}^{m} K_{i}(t) = \sum_{i=1}^{m} K'_{i,t}$.

    \item $\mathcal{B}$ calls $\mathtt{Setup}(1^{\lambda},n)$.
    For $1\leq i\leq m$, when $j \leq \ell$, let
    $$\boldsymbol{tk}_{i}^j :=[\![
     \left(\rho_j,-\frac{\rho_j}{u},0,\beta_{i,j},\gamma_{i,j},0,0,\theta_i\delta_{j,i}\right)R_i]\!]_{2}.$$
    When $j = \ell+1$, let
    $$\boldsymbol{tk}_{i}^{\ell+1} :=[\![\left(0,0,0,0,0,\xi_i,0,0\right)R_i]\!]_{2}.$$
    When $j \geq \ell+2$, let
    $$\boldsymbol{tk}_{i}^j :=[\![
     \left(\rho_j,-\frac{\rho_j}{u},0,\beta_{i,j},\gamma_{i,j},0,\theta_i\delta_{j,\boldsymbol{\pi}(i)},0\right)R_i]\!]_{2}.$$

    \item After receiving the queries $\{x_{i,t}\}_{i\leq m}$ from $\mathcal{A}$, $\mathcal{B}$ computes
    \begin{equation*}
        \boldsymbol{ct}_{i,t} :=[\![C_i \big(
    Y_{i,t},\alpha_{i,t},\alpha'_{i,t},0,0,
    \frac{1}{\xi_i}(\rho_{\ell+1}K_{i,t}+\theta_i\delta_{\ell+1,\boldsymbol{\pi}(i)}x_{i,t})
    ,x_{i,t},x_{\boldsymbol{\pi}^{-1}(i),t}
    \big)^{\top} ]\!]_{1},
    \end{equation*}
    and returns it to $\mathcal{A}$.
    \item $\mathcal{B}$ outputs the identical guess output by $\mathcal{A}$.
\end{itemize}
Observe that if $Y_{i,j} = K_{i,j} + \frac{\alpha_{i,j}}{u}$, then $\mathcal{A}$'s view is identically distributed as in $\textbf{Hyb}^{(2,\ell,2)}$; otherwise, if $Y_{i,j} = K'_{i,j} + \frac{\alpha_{i,j}}{u}$, $\mathcal{A}$'s view is identically distributed as in $\textbf{Hyb}^{(2,\ell,3)}$.
By the XDLin assumption and Proposition \ref{prop_4}, we can obtain that
\begin{equation*}\label{equ_hyb3}
    \left| \mathrm{Pr}\left[ 1 \leftarrow \mathcal{A}\left(\textbf{Hyb}^{(2,\ell,2)}\right) \right] - \mathrm{Pr}\left[ 1 \leftarrow \mathcal{A}\left(\textbf{Hyb}^{(2,\ell,3)}\right) \right] \right| 
     \leq  \frac{1}{q}+Q(m-1)\mathrm{Adv}^{\mathrm{XDLin}}(\lambda).
\end{equation*}

Symmetrically to the steps above, equation (\ref{equ_hyb2l}) can be proved.
Hence, Lemma \ref{lem_hyb2} holds.

}

\bibliographystyle{IEEEtran}
\bibliography{IEEEabrv,SGroupDefinition,ref}

\end{document}